\documentclass[draftcls,one column,journal]{IEEEtran}

\usepackage{ifpdf}
\usepackage{color}

\usepackage{subfigure}

\usepackage{cite}

\ifCLASSINFOpdf
  \usepackage[pdftex]{graphicx}
\else
\fi
\usepackage{epstopdf}
\usepackage{epsfig}

\usepackage[cmex10]{amsmath}
\usepackage{amsthm}
\usepackage{amssymb}
\usepackage{array}

\usepackage[font=footnotesize]{subfig}
\usepackage{url}

\newtheorem{defn}{Definition}
\newtheorem{obs}{Observation}
\newtheorem{thm}{Theorem}
\newtheorem{lem}{Lemma}
\newtheorem{rmk}{Remark}
\newtheorem{cor}{Corollary}
\begin{document}

\title{Quality Sensitive Price Competition in Spectrum Oligopoly-Part 1}

\author{Arnob~Ghosh and
        Saswati~Sarkar
\thanks{The authors are with the Department
of Electrical and Systems Engineering, University Of Pennsylvania, Philadelphia,
PA, USA. Their E-mail ids are arnob@seas.upenn.edu and swati@seas.upenn.edu.} 
\thanks{Parts of this paper have been presented in ISIT\rq{}2013, Istanbul\cite{isit}.}}




\maketitle
\begin{abstract}
We investigate a spectrum oligopoly market where each primary seeks to sell its idle channel to a secondary. Transmission rate of a channel evolves randomly. Each primary needs to select a price depending on the transmission rate of its channel. Each secondary selects a channel depending on the price and the transmission rate of the channel. We formulate the above problem as a non-cooperative game. We show that there exists a unique Nash Equilibrium (NE) and explicitly compute it. Under the NE strategy profile a primary prices its channel to render the channel which provides high transmission rate more preferable; this negates the perception that prices ought to be selected to render channels equally preferable to the secondary regardless of their transmission rates.We show the loss of revenue in the asymptotic limit due to the non co-operation of primaries. In the repeated version of the game, we characterize a subgame perfect NE where a primary can attain a payoff arbitrarily close to the payoff it would obtain when primaries co-operate. 
\end{abstract}
\vspace{-0.3cm}
\section{Introduction}
To accommodate the ever increasing traffic in wireless spectrum, FCC has legalized the unlicensed access of a part of the licensed spectrum band- which is known as {\em secondary access}. However, secondary access will only proliferate when it is rendered profitable to license holders (primaries)\footnote{Primaries are unlikely to encourage \lq\lq{}free ride\rq\rq{} where secondaries (unlicensed users) can access idle channel without paying for them. If primaries are not compensated, they may therefore create artificial activity by transmitting some noise to deny secondaries.}. We devise a framework to enable primaries to decide price they would charge unlicensed users (secondaries). The price will depend on the transmission rate which evolves randomly owing to usage of subscribers of primaries and fading. A secondary receives a payoff from a channel depending on the transmission rate offered by the channel and the price quoted by the primary. Secondaries buy those channels which give them the highest payoff, which leads to a {\em competition} among primaries. 

Price selection in oligopolies has been extensively investigated in economics as a non co-operative Bertrand Game \cite{mwg} and its modifications \cite{Osborne, Kreps}. Price competition among wireless service providers have also been explored to a great extent \cite{Ileri, Mailespectrumsharing, Mailepricecompslotted, Xing, Niyatospeccrn, Niyatomultipleseller,Zhou,kavurmacioglu,yitan,duan,zhang,jia,yang,sengupta,kim}. We  divide this genre of works in two parts: i) Papers which model price competition as Auction (\cite{sengupta,Xu}), and ii) Papers which model the price competition as a non co-operative game (\cite{Ileri, Mailespectrumsharing, Mailepricecompslotted, Xing, Niyatospeccrn, Niyatomultipleseller, kavurmacioglu,jia,yang,yitan,zhang,lin,duan,kim}). We now distinguish our work with respect to these papers. As compared to the genre of work  in the first category our model is readily scalable and a central auctioneer is not required. Some papers in the second category \cite{jia,yang,yitan,zhang,duan,kim} considered the quality of primaries as a factor while selecting the price. But all of these papers mentioned above, ignore two important properties which distinguish spectrum
oligopoly from standard oligopolies:  First, a primary selects a price knowing only the transmission rate of its own channel; it is unaware of transmission rates offered by channels of its competitors. Thus, if a primary quotes a high price, it will earn a large profit if it sells its channel, but may not be able to sell at all; on the other hand a low price will enhance the probability of a sale but may also fetch lower profits in the event of a sale. Second, the same spectrum band can be utilized simultaneously at geographically dispersed locations without interference; but the same band can not be utilized simultaneously at interfering locations; this special feature, known as {\em spatial reuse} 
adds another dimension in the strategic interaction
as now a primary has to cull a set of non-interfering locations which is denoted as an {\em independent set}; at which to offer its channel apart from selecting a price at every node of that set. We have accommodated both of the above the uncertainty of competition in this paper. Building on the results we obtain in this paper, we have accommodated both of the above characteristics in the sequel.

 Some recent works (\cite{Gaurav1,Janssen,Kimmel,gauravjsac}) that consider uncertainty of competition and the impact of spatial reuse (\cite{Zhou} and \cite{gauravjsac}) assume that the commodity on sale can be in one of two states: available or otherwise.  This assumption does not capture different transmission rates offered by available channels. A primary may now need to employ different pricing strategies and different independent set selection strategies for different transmission rates, while in the former case a single pricing and independent set selection strategy will suffice as a price needs not be quoted for an unavailable commodity. Our investigation seeks to contribute in this space.

We have studied the price competition in spectrum oligopoly in a sequel of two papers. Overall a primary has to select an independent set and a pricing strategy in each location of the independent set. In this paper we focus only on the pricing strategy of primaries by studying the game for only one location. The results that we have characterized provide indispensable tools for studying the joint decision problem involving multiple locations. In the sequel to this paper, we provide a framework for solving the joint decision problem building on the single location pricing strategy. In addition, initially secondary market is likely to be introduced in geographically dispersed locations; which are unlikely to interfere with each other, thus, the price competition in each location reduces to a single location pricing problem which we solve in this paper.  The reduced problem turns out to be of independent interest as it exhibits certain desirable properties which no longer hold when there are multiple locations. 

We formulate the price selection as a game in which each primary selects a price depending on the transmission rate its channel provides. Since prices can take real values, the strategy space is uncountably infinite; which precludes the guarantee of the existence of an Nash Equilibrium (NE) strategy profile.  Also standard algorithms for computing an NE strategy profile do not exist unlike when the strategy space is finite. 

First, we consider that primaries interact only once (Section~\ref{sec:onenodegame}). We consider that the preference of the secondaries can be captured by a penalty function which associates a penalty value to each channel that is available for sale depending on its transmission rate and the price quoted. We show that for a large class of penalty functions, there exists a {\em unique} NE strategy profile, which we explicitly compute (Section~\ref{sec:computation}). In the sequel to this paper, we show that this is no longer the case when a primary owns a channel over multiple locations.

  We show that the unique NE strategy profile is symmetric i.e. price selection strategy of all primaries are statistically identical. Our analysis reveals that primaries select price in a manner such that the preference order of transmission rates is retained. This negates the intuition that prices ought to be selected so as to render all transmission rates equally preferable to a secondary. The analysis also reveals that the unique NE strategy profile consists of "nice" cumulative distributions in that they are continuous and strictly increasing; the former rules out pure strategy NEs and the latter ensures that the support sets are contiguous. 
  
  Subsequently, utilizing the explicit computation algorithm for the symmetric NE strategies, %
we analytically investigate the reduction in expected profit suffered under the unique symmetric NE pricing strategies as compared to the maximum possible value allowing for collusion among primaries (Section~\ref{sec:slnumerical}).Finally, we extend our one shot game at single location, to a repeated game where primaries interact with each other multiple number of times (Section~\ref{sec:repeatedgame}) and compute a subgame perfect Nash equilibrium (SPNE) in which a primary attains a payoff which is arbitrarily close to the payoff that a primary would have obtained if primaries were select price jointly; thus, price competition does not lower payoff. 

{\em All proofs are relegated to Appendix}.




\section{System Model} 
\label{sec:model}
  We consider a spectrum market with $l (l\geq 2)$ primaries. Each primary owns a channel. Different channels leased by primaries to secondaries constitute disjoint frequency bands. There are $m$ secondaries. We initially consider the case when primaries know $m$, later generalize our results for random, apriori unknown $m$ (Section~\ref{sec:random}). 
The channel of a primary provides a certain transmission rate to a secondary who is granted access. Transmission rate  (i.e. Shanon Capacity) depends on  1) the number of subscribers of a primary that are using the channel\footnote{Shanon Capacity \cite{cover} for user $i$ at a channel is equal to $\log\left(1+\dfrac{p_{i}h_i}{\sum_{j\neq i}p_jh_j+\sigma^2}\right)$ where $p_k$ is the power with which user $k$ is transmitting, $\sigma^2$ is the power of white noise, $h_k$ is the channel gain between transmitter and receiver which depends on the propagation condition. If a secondary is using the channel then $p_i, h_i$ of the numerator are the attributes associated with the secondary while $p_j, h_j j\neq i$ are those of the subscribers of the primaries. In general, the power $p_j$ for subscriber of primaries is constant for subscriber $j$ of primary, but the number of subscribers vary randomly over time. The power $p_i$ with which a secondary will transmit may be a constant or may decrease with the number of subscribers of primaries in order to limit the interference caused to each subscriber. The above factors contributes to the random fluctuation in the capacity of a channel offered to a secondary. In our setting $p_i, h_i$s are assumed to be the same across the secondaries for a channel which we justify later. However, these values can be different for different channels.} and 2) the propagation condition of the radio signal. The transmission rate evolves randomly over time owing to the random fluctuations of the usage of subscribers of primaries and the propagation condition\footnote{Referring to footnote 2, $h_k$ and $\sigma^2$ evolve randomly owing to the random scattering of the particles in the ionosphere and troposphere; this phenomenon is also known as {\em fading}.}. We assume that at every time slot, the channel of a primary belongs to one of the states $0, 1, \ldots, n$\footnote{We discretize the available transmission rates into a fixed number of states $n$. This is a standard approximation to discretize the continuous function\cite{fischer, Luo}. The corresponding inaccuracy becomes negligible with increase in $n$.}. State $i$ provides a lower transmission rate to a secondary than state $j$ if $i < j$ and state $0$ arises when
  the channel is not available for sale i.e. secondaries can not use the channel when it is in state $0$\footnote{Generally a minimum transmission rate is required to send data. State $0$ indicates that the transmission rate is below that threshold due to either the excessive usage of subscribers of primaries or the transmission condition.}. A channel is in state $i \geq 1 $ w.p. $q_i>0$ and in state $0$ w.p. $1-q$ where $q = \sum_{i=1}^n q_i$, independent of the channel states of other primaries \footnote{We have shown  that at a given slot, the channel state differs across the primaries mainly because of the differences of  i) the number of subscribers that are using the channel and ii) the propagation conditions.  Since different primaries have different subscriber bases, thus, their usage behaviors are largely uncorrelated. Also,  channels of different primaries operate on different frequency bands and have different noise levels, thus, the propagation conditions are also uncorrelated across the channels.}. Thus, the state of the channel of each primary is independent and identically distributed. We do not make any assumption on the relationship between $q_i$ and $l$ or $q_i$ and $m$.\footnote{ Since each primary sells its channel to only one secondary, thus, referring to footnotes 2 and 3 transmission rate (or $q_i$) at a channel does not depend on $m$ (secondary demand) in practice.}.  
 We assume
\begin{align}\label{prob}
q<1.
\end{align}
 
 We assume that the transmission rate offered by the channel of a primary is the same to all secondaries.  We justify the above assumption in the following. We consider the setting where the secondaries are one of the following types: i) Service provider who does not lease spectrum from the FCC and serves the end-users through secondary access, ii) end-users who directly buy a channel from primaries.  In initial stages of deployment of the secondary market, secondaries will be of the first type. When the secondaries are of the first type, then a primary would not know the transmission rate to the end-users who are subscribers of the service provider. A primary measures the channel qualities across different positions in the locality (e.g. a cell) and considers the average as the channel quality that an end-user subscribed to a secondary service provider will get at the location (e.g. a cell). This average will be identical across different end-users subscribed to different secondary service providers and hence, the channel quality is identical across the secondaries.  

If the secondaries are of the second type, then the primary may know the transmission rate that an end-user will attain which may be different at different positions. However, if a primary needs to select a price for each position of the location, then it needs to compute the transmission rate at each possible position at that location (e.g. a cell). Thus, the computation and storage requirement for the primary would be large. Such position based pricing scheme will also not be attractive to the end-users since they may perceive it discriminatory as the price changes when its position changes within a location (e.g. a cell).  Thus, such a position based pricing scheme may not be practically implementable. Hence, a primary estimates the channel quality and decides the price for the estimated channel quality by considering that the channel quality will not significantly vary across the location. This is because  end-users who are interested to buy the channel from a primary at a location most likely have similar propagation paths: for example, secondary users who buy the channels are most often present in buildings (e.g. shopping complex, an office or residential area). The distance from the base station of the primary to the end-users is also similar because the end-users are close to each other in a location. Thus, the path loss component will also similar. Hence, the channel quality is considered to be identical across the secondaries. 

Though the quality of a channel  is identical for secondaries, the quality can vary across the channels. A primary can get an estimate of the transmission rate by sending a pilot test signal at different positions with the location and then, applying some  standard estimation techniques\cite{estimation}.
 
\subsection{Penalty functions of Secondaries and Strategy of Primaries}
Each primary selects a price for its channel if it is available for sale.
We formulate the decision problem of primaries as a non-cooperative game. A primary selects a price with the knowledge of the state of its channel, but without knowing the states of the other channels; a primary however knows $l, m, n, q_1, \ldots, q_n$. 

Secondaries are passive entities. They select channels depending on the price and the transmission rate a channel offers. We assume that the preference of secondaries can be represented by a penalty function. If a primary selects a price $p$ at channel state $i$, then the channel incurs a penalty $g_i(p)$ for all secondaries. As the name suggests, a secondary prefers a channel with a lower penalty. Since lower prices should induce lower penalty, thus, we assume that each $g_i(\cdot)$ is strictly increasing; therefore, $g_i(\cdot)$ is invertible. A primary selects a price for its available channel, but, since there is an one-to-one relation between the price and penalty at each state, we can equivalently consider that primaries select penalties instead. For a given price, a channel of higher transmission rate must induce lower penalty, thus, $g_i(p)<g_j(p)$ if $i>j$. We can also consider that desirability of a channel for a secondary is the negative of penalty. A secondary will not buy any channel whose desirability falls below a certain threshold, equivalently, whose penalty exceeds a certain threshold. We consider that such threshold is the same for each secondary and we denote it as $v$ i.e. no secondary will buy any channel whose penalty exceeds $v$.  Secondaries have the same penalty function and the same upper bound for penalty value ($v$), thus, secondaries are statistically identical. 

Primary $i$ chooses its penalty using an arbitrary probability distribution function (d.f.) $\psi_{i,j}(\cdot)$ \footnote{Probability distribution refers cumulative distribution function (c.d.f.). C.d.f. of a random variable $X$ is the function
$G(x), G(x)=P(X\leq x) \quad \forall x\in \Re$ \cite{df}.}
when its channel is in state $j \geq 1$. If $j = 0$ (i.e., the channel is unavailable), primary $i$ chooses a penalty of $v+1$: this is equivalent
 to considering that such a channel is not offered for sale as no secondary
   buys a channel whose penalty exceeds $v$. 
\begin{defn}
A strategy of a primary $i$ for state $j\geq 1$, $\psi_{i,j}(\cdot)$ provides  the penalty distribution it uses at each node, when the channel state is $j\geq 1$. $S_i=(\psi_{i,1},....,\psi_{i,n})$ denotes the strategy of primary $i$, and $(S_1,...,S_l)$ denotes the strategy profile of all primaries (players).
$S_{-i}$ denotes the strategy profile of primaries other than $i$.
\end{defn}

\subsection{Payoff of Primaries}
We denote $f_i(\cdot)$ as the inverse of $g_i(\cdot)$. Thus, $f_i(x)$ denotes the price when the penalty is $x$ at channel state $i$. We assume that $g_i(\cdot)$ is continuous, thus $f_i(\cdot)$ is continuous and strictly increasing.  Also, $f_i(x)<f_j(x)$ for each $x$ and $i<j$.   Each primary  incurs a transition cost $c>0$ when the primary leases its channel to a secondary, and therefore never selects a price lower than $c$.

If primary $i$ selects a penalty $x$ when its channel state is $j$, then its  payoff is \\
\begin{eqnarray}\label{eq:uij}
\begin{cases} f_j(x)-c  & \text{if the primary sells its channel}\nonumber\\
0 & \text{otherwise. }
\end{cases}
\end{eqnarray}
  Note that if $Y$ is the number of channels offered for sale, for which the penalties are upper bounded by $v$, then
those with $\min(Y, m)$ lowest penalties are sold  since secondaries select channels in the increasing order of penalties. The ties among channels with identical penalties are broken randomly and symmetrically
   among the primaries. 
\begin{defn}\label{defn:expectedpayoff}
$u_{i,j}(\psi_{i,j},S_{-i})$ is the expected payoff when primary $i$\rq{}s channel is at state $j$ and selects strategy $\psi_{i,j}(\cdot)$ and other primaries use strategy $S_{-i}$.
\end{defn}
\subsection{Nash Equilibrium}
We use Nash Equilibrium (NE) as a solution concept which we define below
\begin{defn}\cite{mwg}
A \emph{Nash  equilibrium}  $(S_1, \ldots, S_n)$ is a strategy profile such that no primary can improve its expected profit by unilaterally deviating from its strategy.  So, with $S_i=(\psi_{i,1},....,\psi_{i,n})$, $(S_1, \ldots, S_n)$, is  an NE if for each primary $i$ and channel state $j$
\begin{align}
u_{i,j}(\psi_{i,j},S_{-i})\geq u_{i,j}(\tilde{\psi}_{i,j},S_{-i}) \ \forall \ \tilde{\psi}_{i, j}.
\end{align}
An NE $(S_1, \ldots, S_n)$  is a \emph{symmetric NE} if $S_i = S_j$ for all $i, j.$
\end{defn}
An NE is asymmetric if $S_i\neq S_j$ for some $i,j\in \{1,\ldots,l\}$. 


Note that if $m\geq l$, then primaries select the highest penalty $v$ with probability $1$. This is because, when $m\geq l$, then, the channel of a primary will always be sold. Henceforth, we will consider that $m<l$.

\subsection{A Class of Penalty Functions}\label{sec:a class}
Since $g_i(\cdot)$ is strictly increasing in $p$ and $g_i(p)>g_j(p)$ for $i<j$, we focus on penalty functions  of the form $g_i(p)=h_1(p)-h_2(i)$, where $h_1(\cdot)$ and $h_2(\cdot)$ are strictly increasing in their respective arguments. 
Note that $-g_i(p)$ may be considered as a utility that a secondary would get at channel state $i$ when the price is set at $p$. Such utility functions are commonly assumed to be concave \cite{song}; which is  possible only if $g_i(\cdot)$ is convex in $p$ i.e. $h_1(\cdot)$ is convex. It is easy to show that when $g_i(p)=h_1(p)-h_2(i)$ and $h_1(\cdot)$ is convex, satisfies the following property: 

\textbf{Assumption 1}
\begin{align}\label{con1}
\dfrac{f_i(y)-c}{f_j(y)-c}<\dfrac{f_i(x)-c}{f_j(x)-c}  \ \mbox{ for all } x>y > g_i(c), i<j.
\end{align}
Moreover, when $g_i(p)=h_1(p)/h_2(i)$, then, the inequality in (\ref{con1}) is satisfied for some certain convex functions $h_1(\cdot)$ like $h_1(p)=p^r (r\geq 1),\exp(p)$.  In addition, there is also a large set of functions that satisfy (\ref{con1}), such as: $g_i(p) = \zeta\left(p -h_2(i)\right), g_i(p) = \zeta\left(p/h_2(i)\right)$ where
 $\zeta(\cdot)$ is continuous and strictly increasing. Moreover,
 $g_i(\cdot)$ are such that the inverses are of the form $f_i(x)=h(x)+h_2(i), f_i(x) = h(x)*h_2(i)$, where $h(\cdot)$ is {\em any} strictly increasing function, satisfy Assumption 1. Henceforth, we only consider penalty functions which satisfy (\ref{con1}).

We mostly consider $g_i(\cdot)$s which do not depend on $l,m,n, q_1,\ldots,q_n$ (e.g. Fig.~\ref{fig:dist},~\ref{fig:asymptotic_rne},~\ref{fig:efficiency}). But  in some case we also consider that $g_i(\cdot)$ is a function of $n$ (e.g. Fig.~\ref{fig:example_nvaries}).

\section{One Shot Game: Structure, Computation, Uniqueness and Existence of NE}
\label{sec:onenodegame}
 
First, we identify key structural properties of a NE (should it exist).
  Next we show that the above properties lead to a unique strategy profile which we explicitly compute  - thus the symmetric NE is unique should it exist.  We finally prove that the strategy profile resulting from the structural properties above is   indeed a NE thereby establishing the existence. 
  \vspace{-0.5cm}
\subsection{Structure of an NE}\label{sec:singlestructure}
 We provide some important properties that any NE $\left(S_1, \ldots, S_l\right)$ (where $S_i=\{\psi_{i,1},\ldots,\psi_{i,n}\}$) must satisfy. First, we show that each primary must use the same strategy profile (Theorem~\ref{thm:noasymmetricNE}). In the sequel to this paper, we show that this is no longer the case when there are multiple locations. In fact we show that there may be multiple asymmetric NEs when there are multiple locations. We show that under the NE strategy profile a primary selects lower values of the penalties when the channel quality is high (Theorem~\ref{thm2}). We have also shown that $\psi_{i,j}(\cdot)$ are continuous and contiguous (Theorem~\ref{thm1} and \ref{thm3}). 
\begin{thm}\label{thm:noasymmetricNE}
Each primary must use the same strategy i.e. $\psi_{i,j}(\cdot)=\psi_{k,j}(\cdot)$ $\forall i,k\in \{1,\ldots,l\}$ and $j\in \{1,\ldots,n\}$.
\end{thm}
Theorem~\ref{thm:noasymmetricNE} implies that an NE strategy profile can not be asymmetric. Since channel statistics are identical and payoff functions are identical to each primary, thus this game is symmetric. Given that the game is symmetric, apparently there should only be symmetric NE strategies. Although there are symmetric games where asymmetric NEs do exist \cite{Fey}, we are  able to rule that out in our setting   using the assumptions that naturally arise in practice namely those which are stated in Sections II.A. and II.B and Assumption 1 which is satisfied by a large class of functions that are likely to arise in practice (Section II.E). Thus, a significance of Theorem~\ref{thm:noasymmetricNE} is that Theorem 1 holds for a large class of penalty functions which are likely to arise in practice. However, we show that there may exist asymmetric NE in absence of Assumption 1 (Section~\ref{sec:asymmetricNE}).

Now, we point out another significance of the above theorem.   In a symmetric game it is difficult to implement an asymmetric NE. For example, for two players if $(S_1,S_2)$ is an NE with $S_1\neq S_2$, then $(S_2,S_1)$ is also an NE due to the symmetry of the game. The realization of such an NE is only possible when one player knows whether the other  is using $S_1$ or $S_2$. But,  coordination among players is infeasible apriori as the game is non co-operative. Thus, Theorem~\ref{thm:noasymmetricNE} entails that we can avoid such complications in this game. We show that this game has a unique symmetric NE through Lemma~\ref{lm:computation} and Theorem~\ref{thm4}. Thus, Theorem~\ref{thm:noasymmetricNE}, Lemma~\ref{lm:computation} and Theorem~\ref{thm4} together entail that there exists a unique NE strategy profile.

Since every primary uses the same strategy, thus, \textbf{we drop the indices corresponding to primaries and  represent the strategy of any primary as $S = (\psi_{1}(\cdot), \psi_2(\cdot),.....,\psi_n(\cdot)) $, where $\psi_{i}(\cdot)$ denotes the strategy at channel state $i$.} Thus, we can represent a strategy profile in terms of only one primary.


\begin{defn}\label{defn:phi}
$\phi_j(x)$ is the expected profit of a primary whose channel is in state $j$ and selects
  a penalty $x$ \footnote{Note that $\phi_j(x)$ depend on strategies of other primaries , to keep notational simplicity, we do not make it explicit}.
  \end{defn}
\begin{thm}
\label{thm1}
$\psi_{i}(.), i\in \{1,..,n\}$ is a continuous probability distribution. Function $\phi_j(\cdot)$ is continuous.
\end{thm}
The above theorem implies that $\psi_i(\cdot)$ does not have any jump at any penalty value. i.e. no penalty value is chosen with positive probability. We now intuitively justify the property. There are uncountably infinite number of penalty values and thus, clearly $\psi_i(\cdot)$ can only have jump at some of those values. Intuitively, there is no inherent asymmetry amongst the penalty values within the interval $(g_i(c),v)$ i.e. at the penalty values except the end points of the interval $[g_i(c),v]$. Thus, a primary does not prioritize any of those penalty values. Now, we intuitively explain why $\psi_i(\cdot)$ does not have jump at the end points. First, at penalty $g_i(c)$, a primary gets a payoff of $0$ when the channel state is $i$; but the payoff at any penalty value greater that $g_i(c)$ is positive, thus $\psi_i(\cdot)$ does not prioritize the penalty value $g_i(c)$.  On the other hand, intuitively if a primary selects penalty $v$ with positive probability, then the rest would select slightly lower penalty in order to enhance their sales and thus, the probability that the primary would sell its channel decreases. Thus, $\psi_i(\cdot)$ also does not have a jump at $v$.

Note that in a deterministic N.E. strategy at channel state $i$, then $\psi_{i}(\cdot)$  must have a jump from $0$ to $1$ at the above penalty value. Such $\psi_i(\cdot)$ is {\em not} continuous. Thus, the above theorem rules out any deterministic N.E. strategy.  The fact that $\phi_j(\cdot)$ is continuous has an important technical consequence; this guarantees the existence of the best response penalty in Definition~\ref{br} stated in Section~\ref{sec:computation}.

\begin{defn}\label{dlu}
 We denote the lower and upper endpoints of the support set\footnote{The support set of a probability distribution is the smallest closed set such that the probability of its complement is $0.$\cite{df}}   of $\psi_i(.)$ as $L_i$ and $U_i$ respectively i.e.
\begin{equation*}\label{n77}
L_i=\inf\{x: \psi_{i}(x)>0\}.
\end{equation*}
\begin{equation*}\label{n77a}
U_i=\inf\{x: \psi_i(x)=1\}.
\end{equation*}
\end{defn}

We next show that primaries select higher penalty when the transmission rate is low. More specifically, we show that upper endpoint of the support set of $\psi_{i}(\cdot)$ is upper bounded by the lower endpoint of $\psi_{i-1}(\cdot)$.
\begin{thm}
\label{thm2}
$U_i\leq L_{i-1}$, if $j<i$.
\end{thm}
Theorem~\ref{thm2} is apparently counter intuitive. We prove it using the assumptions stated in Section II. In particular, we rely on Assumption 1 which is satisfied by a large class of penalty functions (Section~\ref{sec:a class}). Thus, the significance of Theorem~\ref{thm2} is that the counter intuitive structure holds for a large class of penalty functions. However, in Section~\ref{sec:counterexample2} we show that Theorem~\ref{thm2} needs not to hold in absence of Assumption 1.

Fig.~\ref{fig:dist} illustrates $L_i$s and $U_i$s in an example scenario (The distribution $\psi_i(\cdot)$ in Fig.~\ref{fig:dist} is plotted using (\ref{c5})).
\begin{figure*}
\begin{minipage}{.32\linewidth}
\begin{center}
\includegraphics[width=65mm, height=40mm]{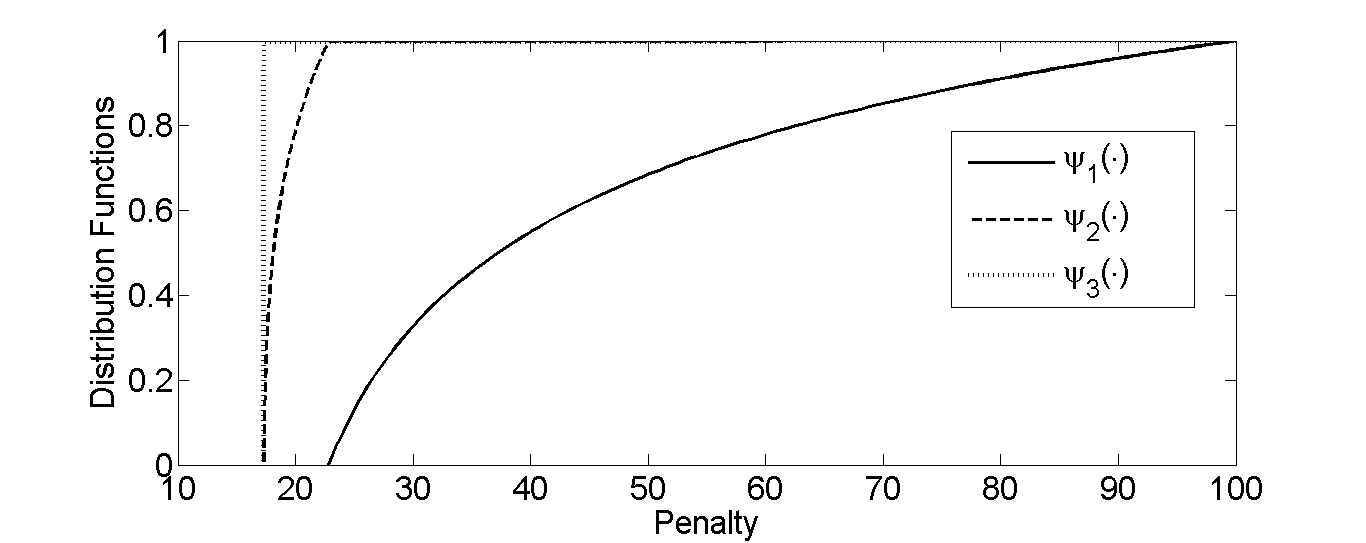}
\vspace*{-.1cm}

\end{center}
\end{minipage}\hfill
\begin{minipage}{.32\linewidth}
\begin{center}
\includegraphics[width=60mm, height=40mm]{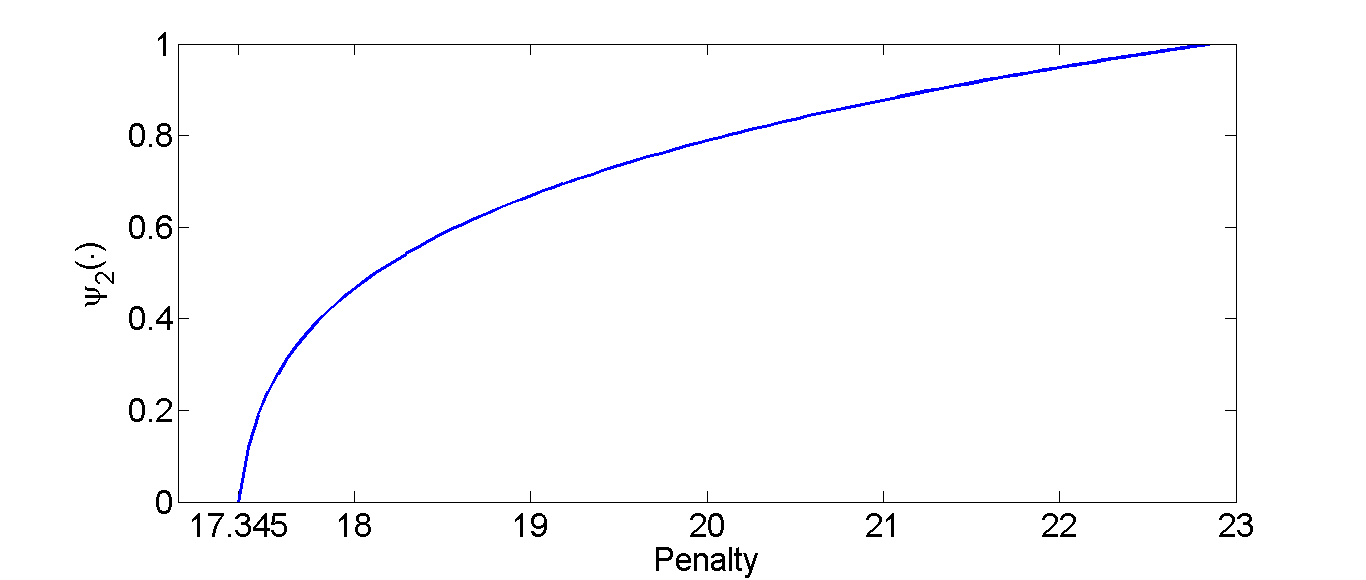}
\end{center}
\end{minipage}\hfill
\begin{minipage}{.32\linewidth}
\begin{center}
\includegraphics[width=60mm, height=40mm]{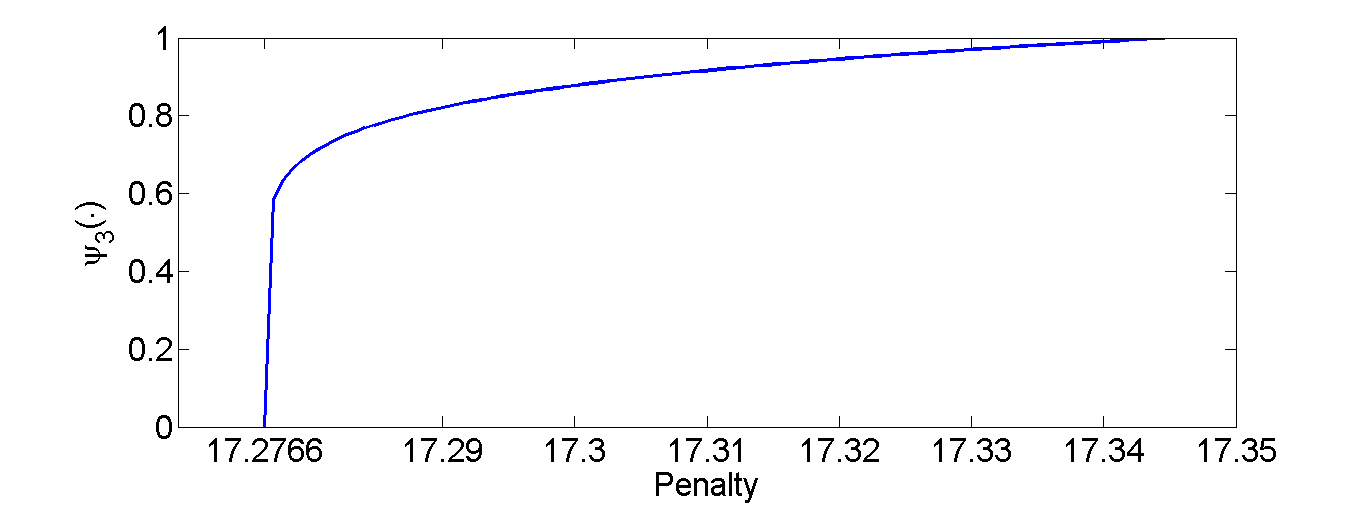}
\end{center}
\end{minipage}
\caption{\small Figure in the left hand side shows the d.f. $\psi_i(\cdot), i=1,\ldots,3$ as a function of penalty for an example setting: $v=100, c=1, l=21, m=10, n=3, q_1=q_2=q_3=0.2$ and $g_i(x)=x-i^3$. Note that support sets of $\psi_i(\cdot)$s are disjoint with $L_3=17.2766$, $U_3=17.345=L_2$, $U_2=22.864=L_1$, and $U_1=100=v$. Figures in the center and the right hand side show d.f. $\psi_2(\cdot)$ and $\psi_3(\cdot)$ respectively, using different scales compared to the left hand figure.}
\label{fig:dist}
\vspace*{-.6cm}
\end{figure*}

In general, a continuous NE penalty selection distribution may not be contiguous i.e. support set may be a union of multiple number of disjoint closed intervals. Thus, the support set of $\psi_{i}(\cdot)$ may be of the following form $[a_1,b_1]\cup\ldots\cup [a_d,b_d]$ with $b_k<a_{k+1}, k\in\{1,\ldots,d-1\}, a_1=L_i$ and $b_d=U_i$. In this case, $\psi_i(\cdot)$ is strictly increasing in each of $[a_k,b_k] k\in \{1,\ldots,d\}$, but it is constant in the \lq\lq{}gap\rq\rq{} between the intervals i.e. $\psi_i(\cdot)$ is constant in the interval $[a_{k-1},b_k]$ $k\in\{2,\ldots,d\}$. We rule out the above possibility in the following theorem i.e. the support set of $\psi_i(\cdot)$ consists of only one closed interval $[L_i,U_i]$ which also establishes that $\psi_{i}(\cdot)$ is strictly increasing from $L_i$ to $U_i$. In the following theorem we also rule out any \lq\lq{}gap\rq\rq{} between support sets for different $\psi_i(\cdot)$, $i=1,\ldots,n$. 
\begin{thm}
\label{thm3}
The support set of $\psi_i(\cdot), i=1,..,n$ is $[L_i, U_i]$  and $U_i=L_{i-1}$ for $i=2,..,n$, $U_1=v$.
\end{thm}
Theorem~\ref{thm2} states that $L_{i-1}\geq U_i$. Theorem~\ref{thm3} further confirms that $L_{i-1}=U_{i}$ i.e. there is no \lq\lq{}gap\rq\rq{} between the support sets. Theorem~\ref{thm3} also implies that there is no \lq\lq{}gap\rq\rq{} within a support set. We now explain the intuition that leads to the reason. If there are $a$ and $b$ which are in the support sets such that the primaries do not select any penalty in the interval $(a,b)$, then, a primary can get strictly a higher payoff at any penalty in the interval $(a,b)$ compared to penalty at $a$ or just below $a$. Thus, a primary would select penalties at or just below $a$ with probability $0$ which implies that $a$ can not be in the support set of an NE strategy profile. We prove Theorem~\ref{thm3} using the above insights and Theorem~\ref{thm2}. 

Figure~\ref{fig:dist} illustrates d.f. $\psi_{i}(\cdot)$ for an example scenario. 
\vspace{-0.5cm}
\subsection{Computation, Uniqueness and Existence}\label{sec:computation}
We now show that the structural properties of an NE identified in
Theorems~\ref{thm:noasymmetricNE}-\ref{thm3} are satisfied by a unique strategy profile,
 which we explicitly compute (Lemma~\ref{lm:computation}). This proves the uniqueness of a NE subject to the existence. We show that the strategy profile is indeed an NE in Theorem~\ref{thm4}.
 We start with the following definitions.
 
 \begin{defn}\label{br}
 A \emph{best response} penalty for  a channel in state  $j\geq 1$  is $x$ if and only if
\begin{equation}\label{eq:bestresponse}
\phi_j(x)=\underset{y\in \Re}{\text{sup}}\phi_j(y).
\end{equation}
Let $u_{j,max} = \phi_j(x)$ for a best response $x$\footnote{Since $\phi_j(\cdot)$ is continuous by Theorem~\ref{thm1} and penalty must be selected within the interval $[g_j(c),v]$, thus the maximum exists in (\ref{eq:bestresponse}). This maximum is equal to $u_{j,max}$ and $\phi_j(x)$ is equal to $u_{j,max}$ for some $x$.} for state $j$, $j \geq 1$ i.e.,  $u_{j,max}$ is the maximum expected profit  that a primary earns when its channel is in state  $j$, $j \geq 1$.
\end{defn} 
\begin{rmk}
In an NE strategy profile a primary only selects a best response penalty with a positive probability. Thus, a primary selects $x$ with positive probability at channel state $i$, then expected payoff to the primary must be $u_{i,max}$ at $x$.
\end{rmk}
\begin{defn}\label{defn:w}
Let $w(x)$ ($w_i$, respectively) be the probability of at least $m$ successes out of $l-1$ independent Bernoulli trials, each of which occurs with probability $x$ ($\sum_{k=i}^{n}q_k$, respectively). Thus, 
\begin{eqnarray}
w(x)&=&\sum_{i=m}^{l-1}\dbinom{l-1}{i}x^i(1-x)^{l-i-1}.\label{d4}\\
w_i&=&w\left(\sum_{j=i}^{n}q_j\right) \quad \text{for } i=1,...,n \quad \text{\&} w_{n+1}=0\label{d5}.
\end{eqnarray}
\end{defn}

The following lemma provides the explicit expression for the maximum expected payoff that a primary can get under an NE strategy profile.
\begin{lem} \label{lu} For $1 \leq i \leq n$,
\begin{eqnarray}u_{i,max} & = & p_i-c .\nonumber\\
\mbox{where, } p_i& =& c+(f_i(L_{i-1})-c)(1-w_i).
\label{n51}\\
\mbox{and } L_i&=&g_i(\dfrac{p_i-c}{1-w_{i+1}}+c), L_{0}=v.\label{n52}
\end{eqnarray}
\end{lem}

\begin{rmk}\label{rmk:expectedpayoff}
Expected payoff obtained by a primary under an NE strategy profile at channel state $i$ is given by $p_i-c$.
\end{rmk}
\begin{rmk}
Starting from $L_0=v$, we obtain $p_1$ (from (\ref{n51})) and using $p_1$ we obtain $L_1$ by (\ref{n52}) which we use to obtain $p_2$ from (\ref{n51}). Thus, recursively we obtain $p_i$ and $L_i$ for $i=1,\ldots,n$.
\end{rmk}


 We now obtain expressions for $\psi_{i}(\cdot)$ using expression of $L_i$ and $p_i$. We use the fact that $w(\cdot)$ is strictly increasing and continuous in $[0,1]$.
\begin{lem} \label{lm:computation}
An NE strategy profile (if it exists) $\left(\psi_1(\cdot), \ldots, \psi_n(\cdot)\right)$ must
comprise of: \begin{align}\label{c5}
\psi_i(x)= & 
 0 ,  \text{if} \ x<L_i\nonumber\\
& \dfrac{1}{q_i}(w^{-1}(\dfrac{f_i(x)-p_i}{f_i(x)-c})-\sum_{j=i+1}^{n}q_j), \text{if} \ L_{i-1}\geq x\geq L_i\nonumber\\
& 1,  \text{if} \ x>L_{i-1}.
\end{align}
where $p_i, L_i, i=0,..,n$ are defined in (\ref{n51}).
\end{lem}

\begin{rmk}
Using (\ref{c5}) we can easily compute the strategy profile $(\psi_1(\cdot),\ldots,\psi_{n}(\cdot))$. Fig.~\ref{fig:dist} illustrates d.f. $\psi_i(\cdot)$ for an example scenario.
\end{rmk}
The following lemma ensures that $\psi_i(\cdot)$ as defined in Lemma~\ref{lm:computation} is indeed a strategy profile. 

\begin{lem}\label{dcont}
$\psi_i(\cdot)$ as defined in Lemma~\ref{lm:computation} is a strictly increasing and continuous probability distribution function.
\end{lem}
Fig.~\ref{fig:dist} illustrates continuous and strictly increasing $\psi_i(\cdot)$ for $i=1,\ldots, 3$ for an example setting.


Explicit computation in Lemma~\ref{lm:computation} shows that the NE strategy profile is unique, if it exists. There is a plethora of symmetric games \cite{mwg} where NE strategy profile does not exist. However, we establish that any strategy profile of the form (\ref{c5}) is an NE. 
\begin{thm}\label{thm4} Strategy profile as defined in Lemma~\ref{lm:computation} is  an NE.
\end{thm}

Hence, we have shown that
\begin{thm}\label{singlelocation}
The strategy profile, in which each  primary randomizes over the penalties in the range $[L_i,L_{i-1}]$ using the continuous probability distribution function $\psi_i(\cdot)$  (defined in Lemma~\ref{lm:computation}) when the channel state is $i$, is the unique NE strategy profile.
\end{thm}
\subsection{Random Demand}\label{sec:random}
Note that all our results readily generalize to allow for random number of secondaries (M) with probability mass functions (p.m.f.) $\Pr(M=m)=\gamma_m$. A primary does not have the exact realization of number of secondaries, but it knows the p.m.f. . We only have to redefine  $w(x)$ as-
\begin{eqnarray}
\sum_{k=0}^{\max(M)}\gamma_k\sum_{i=k}^{l-1}\dbinom{l-1}{i}x^i(1-x)^{l-1-i}\nonumber
\end{eqnarray}
and $w_{n+1}=\gamma_0$.


\section{Performance Evaluations of NE Strategy Profile in Asymptotic Limit}\label{sec:slnumerical}
 In this section, we analyze the reduction in payoff under NE strategy profile due to the competition. 
 First, we study the expected payoff that a primary obtains under the unique NE strategy profile in the asymptotic limit (Lemma~\ref{eff}). Subsequently, we compare the expected payoff of primaries under the NE strategy profile with the payoff that primaries get when they collude (Lemma~\ref{thresh}). Subsequently, we investigate the asymptotic variation of strategy profiles of primaries with $n$ in an example setting (Fig.~\ref{fig:example_nvaries}).

Recall from Remark~\ref{rmk:expectedpayoff} that expected payoff obtained by a primary under the unique NE strategy profile at channel state $i$ is given by $p_i-c$. Next,
\begin{defn}\label{payoff}
Let $R_{NE}$ denote  the ex-ante expected profit of a primary at the Nash equilibrium. Then,
\begin{eqnarray}\label{eq:rne}
R_{NE}=\sum_{i=1}^{n}(q_i.(p_i-c)).
\end{eqnarray}
\end{defn}
Note that $l\cdot R_{NE}$ denotes the total ex-ante expected payoff obtained by primaries at the NE strategy profile. We obtain
\begin{lem}\label{eff}
Let $c_{j}=g_j(c), j=1,..,n$. When $l\rightarrow \infty$, then
\begin{eqnarray}
p_i-c\rightarrow \begin{cases} f_i(v)-c\quad & \text{if } (l-1)\sum_{j=1}^{n}q_j<m\nonumber\\
f_i(c_{k})-c \quad \text{if }  & (l-1)\sum_{j=k+1}^{n}q_j<m\nonumber\\ & <(l-1)\sum_{j=k}^{n}q_j \quad 1\leq k< i\nonumber\\
0 \quad \text{if } & m< (l-1)\sum_{j=i}^{n}q_j.
\end{cases}
\end{eqnarray}
\end{lem}
\begin{figure*}
\begin{minipage}{.48\linewidth}
\begin{center}

\includegraphics[width=60mm, height=40mm]{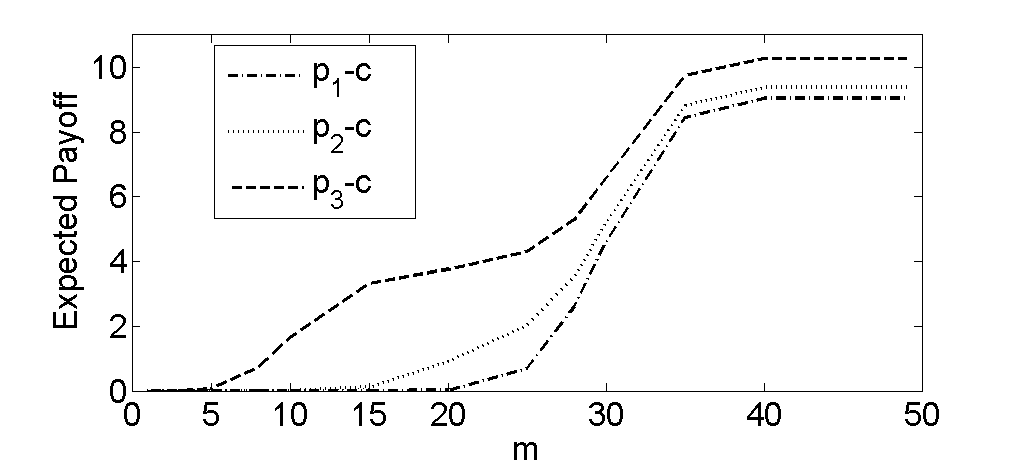}
\caption{\small This figure illustrates the variation of $p_i-c, i=1,\ldots,n$  with $m$ in an example setting: $l=51, n=3, v=100, c=1, q_1=q_2=q_3=0.2$ and $g_i(x)=x^2-i^3$. For $m\leq 5$, $p_i-c\approx 0$ for all $i$. For $5\leq 5\leq 15$, $p_i-c\approx 0$ for $i=1,2.$. For $15\leq m\leq 20$, $p_1-c\approx 0$.  When $m$ exceeds $40$, $p_i-c, i=1,2,3$ closely match the highest possible expected value as Lemma~\ref{eff} indicates.}
\label{fig:asymptotic_rne}
\vspace{-0.4cm}
\end{center}
\end{minipage}\hfill
\begin{minipage}{.48\linewidth}
\begin{center}
\includegraphics[width=60mm,height=40mm]{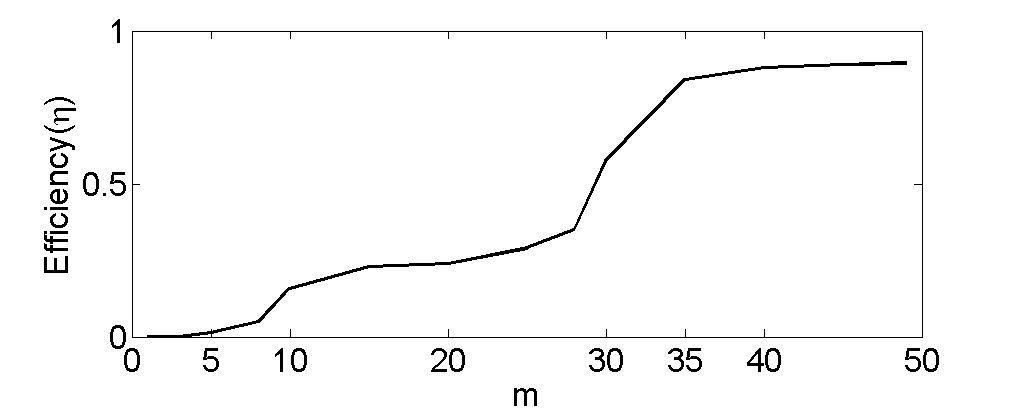}
\caption{\small Variation of efficiency ($\eta$) with $m$ in an example setting: $g_i(x)=x^2-i^3$, $l=51, n=3, q_1=q_2=q_3=0.2, v=100$ and $c=1$. When $m\leq 5$, $\eta\approx0$. When $m\geq 35$, $\eta\approx 1$.}
\label{fig:efficiency}
\vspace{-0.4cm}
\end{center}
\end{minipage}
\end{figure*}

We illustrate Lemma ~\ref{eff} using an example in Figure~\ref{fig:asymptotic_rne}.
Intuitively, competition increases with the decrease in $m$. Primaries choose prices progressively closer to the lower limit $c$. Thus, the expected payoff $p_i-c, i=1\ldots,n$ decreases as $m$ decreases. The above lemma reveals that, as $m$ becomes smaller only those primaries whose channels provide higher transmission rate can have strictly positive payoff i.e. $p_i-c$ is positive (Fig.~\ref{fig:asymptotic_rne}). 

From Lemma~\ref{eff} and (\ref{eq:rne}) we readily obtain
\begin{cor}\label{cor:rne}
When $l\rightarrow \infty$, then,
\begin{eqnarray}
R_{NE}\rightarrow\begin{cases} \sum_{j=1}^{n}q_j.(f_{j}(v)-c)\quad &\mbox{If } (l-1)\sum_{j=1}^n q_j<m\nonumber\\
\sum_{j=i+1}^{n}q_j.(f_{j}(c_{i})-c) \quad &\mbox{If }  (l-1)\sum_{j=i+1}^{n}q_j<m\nonumber\\ & <(l-1)\sum_{j=i}^{n}q_j \nonumber\\ & , i\in\{1,..,n-1\}\nonumber\\
0\quad &\mbox{If } m< (l-1)q_n.\nonumber\\\end{cases}
\end{eqnarray}
\end{cor}
Thus, asymptotically $R_{NE}$ decreases as $m$ decreases (Fig.~\ref{fig:asymptotic_rne}).
\begin{defn}\label{defn:eta}
Let $R_{OPT}$ be  the maximum  expected profit earned through collusive selection of
 prices by the primaries. \emph{Efficiency} $\eta$ is defined as $\dfrac{l\cdot R_{NE}}{R_{OPT}}.$
\end{defn}
Efficiency is a measure of the reduction in the expected profit owing to competition. 
The asymptotic behavior of $\eta$ is characterized by the following lemma.
\begin{lem}\label{thresh}
When $l\rightarrow \infty$, then
\begin{eqnarray}
\eta \rightarrow\begin{cases}
 1 \quad &\mbox{If } (l-1)\sum_{j=1}^{n}q_j<m \nonumber\\
 0 \quad & \mbox{If } m< (l-1)q_n.
         \end{cases}
         \end{eqnarray}
\end{lem}
We illustrate the variation of efficiency with $m$ using an example in Figure~\ref{fig:efficiency}. Intuitively, the competition decreases with increase in  $m$; thus primaries set their penalties close to the highest possible value for all states. This leads to high efficiency.  On the other hand, competition becomes intense when $m$ decreases, thus, $R_{NE}$ becomes very small as Corollary ~\ref{cor:rne} reveals. But, if primaries collude,  primaries can maximize the aggregate payoff by offering only the channels of highest possible states by selecting highest penalties. Thus, efficiency becomes very small when $m$ is very small (Fig.~\ref{fig:efficiency}).

The transmission rates of an available channel constitute a continuum in practice. We have discretized the transmission rates of an available channel in multiple states for the ease of analysis. However, the theory allows us to investigate numerically how the penalty distribution strategies behave in the asymptotic limit (Fig.~\ref{fig:example_nvaries}). 

\begin{figure}
\includegraphics[width=100mm,height=40mm]{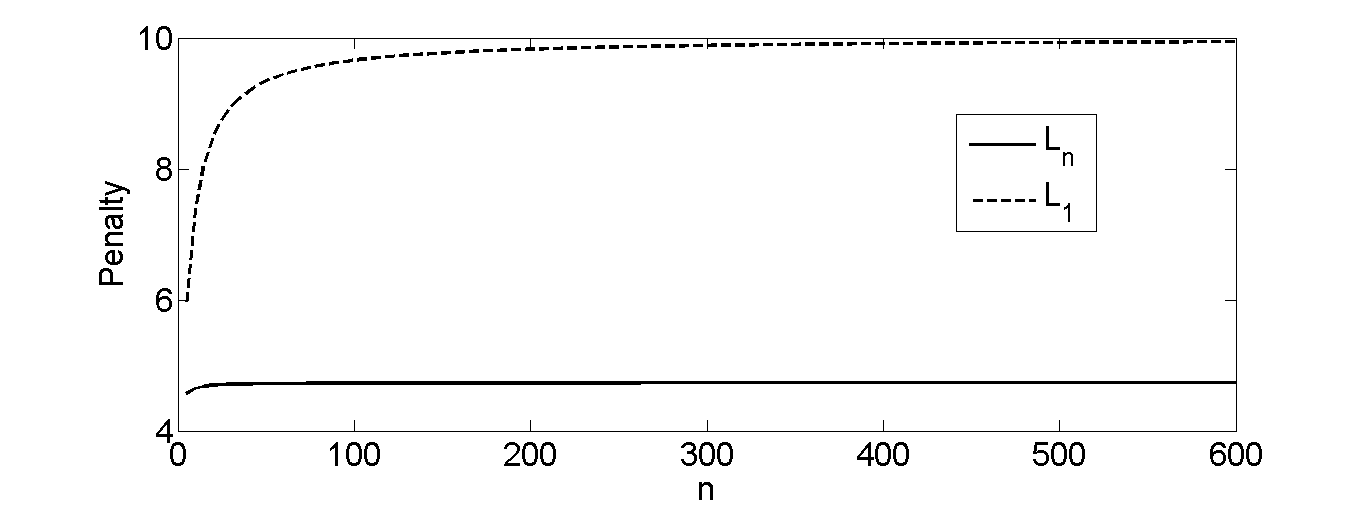}
\caption{\small Figure shows the plot of $L_1$ and $L_n$ with $n$. We consider $v=10, c=0, l=21, m=10, q_1=\ldots=q_n=0.5/n$ and the following penalty function: $g_i(x)=x-(r_{max}-r_{min})*i/n$, where $r_{max}$ denotes the maximum possible transmission rate which a secondary user can transmit and $r_{min}$ denotes the minimum transmission rate required to transmit the signal. Note that the penalty function is of the form $g_i(x)=h_1(x)-h_2(i)$,  where, $h_1(x)=x$ is strictly increasing concave function in $x$, and $h_2(i)=(r_{max}-r_{min})*i/n$ is strictly decreasing in $i$. Thus, $g_i(\cdot)$ satisfies Assumption 1. We have equally divided the available rate region into the number of states $n$. We consider that $h_2(\cdot)$ is the representative rate at state $i$. We consider $r_{max}=3.5$ and $r_{min}=0.5$.}
\label{fig:example_nvaries}
\vspace{-0.4cm}
\end{figure}
Fig.~\ref{fig:example_nvaries} reveals that $L_n$ increases with $n$ and eventually converges to a point which is strictly less than $v$. On the other hand, $L_1$ converges to $v$ as $n$ becomes large (Fig.~\ref{fig:example_nvaries}). Thus, the lower endpoints (and thus, the upper endpoints (since $U_{j}=L_{j-1}$)) of penalty selection strategies converge at different points when $n$ becomes large.


\section{Repeated Game}\label{sec:repeatedgame}

We have so far considered the setting where primaries and secondaries interact only once. In practice, however, primaries and secondaries interact repeatedly. To analyze repeated interactions we consider a scenario where the one shot game is repeated an infinite number of times. We characterize the subgame perfect Nash Equilibrium where the  expected payoff of primaries can be made arbitrarily close to the payoff that primaries would obtain if they would collude.  

The one shot game is played at time slots $t=0,1,2,\ldots,$.  Let, $\phi_{i,t}$ denote the expected payoff at stage $t$, when the channel state is $i$. Hence, the payoff of a primary, when its channel state is $i$, is given by 
\begin{align}
\phi_i=(1-\delta)\sum_{t=0}^{\infty}\delta^{t}\phi_{i,t}
\end{align}
where, $\delta\in(0,1)$ is the discount factor. 

Since NE constitutes a weak notion in repeated game \cite{mwg}, we will focus on Subgame Perfect Nash Equilibrium (SPNE).
\begin{defn}\cite{mwg}
Strategy profile $(S_1,\ldots,S_n)$ constitutes a SPNE if the strategy profile prescribes an NE at every subgame.
\end{defn}
The one shot unique NE that we have characterized, is also a SPNE in repeated game. Total expected payoff that a primary can get, is  $R_{NE}$ (Definition~\ref{payoff}) under one-shot game. We have already shown that this payoff is not always efficient (Lemma~\ref{thresh}) i.e. $\eta\nrightarrow 1$. Here, we present an efficient SPNE (Theorem~\ref{spne}), provided that $\delta$ is sufficiently high.

Fix a sequence  of $\{\epsilon_i\}_{1}^{n}$ such that 
\begin{eqnarray}
 0=\epsilon_1<\epsilon_2<\ldots<\epsilon_n. \label{condn2}\\
\epsilon_i\leq v-L_{i-1},
\epsilon_i\leq v-g_i((f_i(v)-c)(1-w_i)+c).\label{condn3}
\end{eqnarray}
We provide a {\em Nash Reversion Strategy} such that a primary will get a payoff arbitrarily close to the payoff that it would obtain when all primaries collude.\\
\textbf{Strategy Profile ($SP_R$)}:
\textit{ Each primary selects penalty $v-\epsilon_i$, (where $\epsilon_i$ satisfies (\ref{condn2}) and (\ref{condn3})), when state of the channel is $i$,  at $t=0$ and also at time slot $\tau=1,2,\ldots$ as long as all other primaries have chosen $v-\epsilon_j$, $j\in \{1,\ldots,n\}$, when their channel state is $j$ at all time slots $0,1,\ldots,\tau-1$. Otherwise, play the unique one shot game NE strategy $\psi_i(\cdot)$ (Lemma~\ref{lm:computation}).}

\begin{rmk}
Note that if everyone sticks to the strategy, then each primary selects a penalty of $v-\epsilon_i$ at every time slot, when the channel state is $i$. Under the collusive setting, each primary selects penalty $v$. Thus, for every $\gamma>0$, we can choose sufficiently small $\epsilon_i, i=1,\ldots,n$ and sufficiently high $\delta$ such that the efficiency (definition~\ref{defn:eta}) is at least $1-\gamma$.
\end{rmk}

Now we are ready to state the main result of this section.
\begin{thm}\label{spne}
Suppose $\{\epsilon_i\}_{i=1}^{n}$ are such that they satisfy (\ref{condn2}) and (\ref{condn3}). Then, there exists $\delta_{min}\in (0,1)$ such that for any discount factor $\delta\geq \delta_{min}$ the strategy profile $SP_R$ is a SPNE.
\end{thm}
\begin{rmk}
Thus, there exists a SPNE strategy profile (for sufficiently high $\delta$)  where each primary obtains an expected payoff arbitrarily close to the payoff it would have obtained if all primaries collude.
\end{rmk}
\section{Pending Question: What happens when Assumption 1 is relaxed?}\label{sec:assumpnecessary}
We show that if penalty functions do not satisfy Assumption 1, then the system may have multiple NEs (section~\ref{sec:multipleNEs}), asymmetric NE (section~\ref{sec:asymmetricNE}) and the strategy profile that we have characterized in (\ref{c5}) may not be an NE (section~\ref{sec:counterexample2}).
\subsection{Multiple NEs}\label{sec:multipleNEs}
We first give a set of penalty functions which do not satisfy Assumption 1 and then we state a strategy profile which is an NE for this set of penalties along with the strategy profile that we have characterized in (\ref{c5}). \\
Let $f_i(\cdot)$ be such that 
\begin{align}\label{mt1}
\dfrac{f_i(x)-c}{f_j(x)-c}=\dfrac{f_i(y)-c}{f_j(y)-c}\quad (x>y>g_i(c), i<j)
\end{align}
Examples of such kind of functions are $g_i(x)=(x-c)^p/i$. 

It can be easily verified that strategy profile, described as in (\ref{c5}), is still an NE strategy profile under the above setting. We will provide another NE strategy profile. 

First, we will introduce some notations which will be used throughout this section. 
\begin{eqnarray}
\bar{p}_i=(f_i(v)-c)(1-w_1)+c \label{mt3}\\
\bar{L}=g_1(\bar{p}_1)\label{mt4}
\end{eqnarray}

Now, we show that there exists a symmetric NE strategy profile where a primary selects the same strategy for its each state of the channel. This establishes that the the system has multiple NEs.

Let\rq{}s consider the following symmetric strategy profile where at channel state $i$ a primary\rq{}s strategy profile is $\bar{\psi}_i(\cdot)=\bar{\psi}(\cdot)$ for $i=1,\ldots,n$, where 
\begin{align}\label{dfmt}
\bar{\psi}(x)=& 0 \quad(\text{if } x<\bar{L})\nonumber\\
& \dfrac{1}{\sum_{j=1}^{n}q_j}w^{-1}(1-\dfrac{\bar{p}_1-c}{f_1(x)-c})\quad (\text{if } v\geq x\geq \bar{L})\nonumber\\
& 1\quad (\text{if } x>v)
\end{align}
First, we show that $\bar{\psi}(\cdot)$ is a probability d.f.
\begin{lem}\label{lm:strategymultiple}
$\bar{\psi}(\cdot)$ as defined in (\ref{dfmt}) is a probability distribution function.
\end{lem} 

Note that in this strategy profile each primary selects the same strategy irrespective of the channel state. Next, we show that strategy profile as described in (\ref{dfmt}) is an NE strategy profile. 
\begin{thm}\label{thm:multipleNE}
Consider the strategy profile where $\bar{\psi}_i(\cdot)=\bar{\psi}(\cdot)$, for $i=1,\ldots, n$. This strategy profile constitute an NE.
\end{thm}
 
\subsection{Asymmetric NE}\label{sec:asymmetricNE}
If Assumption 1 is not satisfied, then there may exist asymmetric NEs in contrast to the setting when Assumption 1 is satisfied. We again consider the penalty functions are of the type given in (\ref{mt1}). 
We consider $n=2, l=2, m=1$ and $q_1=q_2$.
Here, we denote $\psi_{i,j}(\cdot)$ as the strategy profile for primary $i , i=1,2 $ at channel state $j$, $j=1,2$. Let 
\begin{align}\label{eq:asymce1}
\hat{L}=g_2((f_2(v)-c)(1-q_1-q_2)/(1-q_2)+c)
\end{align}
Note from (\ref{eq:asymce1}) that 
\begin{align}
(f_2(\hat{L})-c)(1-q_2)=(f_2(v)-c)(1-q_1-q_2)\nonumber\\
(f_1(\hat{L})-c)(1-q_2)=(f_1(v)-c)(1-q_1-q_2) \quad (\text{from } (\ref{mt1}))\nonumber\\
(f_1(\hat{L})-c)(1-q_1)=(f_1(v)-c)(1-q_1-q_2)\quad (\text{since } q_1=q_2)
\end{align}
Next 
\begin{align}
\hat{L}_{low}=g_1((f_1(\hat{L})-c)(1-q_2)+c)
\end{align}
Again using (\ref{mt1}) and the fact that $q_1=q_2$ we also obtain
\begin{align}
f_2(\hat{L}_{low})-c=(f_2(\hat{L})-c)(1-q_1)
\end{align}
Consider the following strategy profile
\begin{align}\label{eq:asymmne1}
\psi_{1,1}(x)& =1,\quad (x\geq v),\nonumber\\
& =\dfrac{1}{q_1}(1-q_2-\dfrac{(f_2(v)-c)(1-q_1-q_2)}{f_2(x)-c})\quad (v>x>\hat{L})\nonumber\\
& =0, \quad x\leq \hat{L}\nonumber\\
\psi_{1,2}(x)& =1, \quad (x\geq \hat{L}), \nonumber\\
& =\dfrac{1}{q_2}(1-\dfrac{(f_1(\hat{L})-c)(1-q_2)}{f_1(x)-c})\quad (\hat{L}>x>\hat{L}_{low})\nonumber\\
& =0,\quad x\leq \hat{L}_{low}
\end{align}
and
\begin{align}\label{eq:asymmne2}
\psi_{2,1}(x)& =1,\quad (x\geq \hat{L}),\nonumber\\
& =\dfrac{1}{q_1}(1-\dfrac{(f_2(\hat{L})-c)(1-q_1)}{f_2(x)-c})\quad \hat{L}_{low}<x<\hat{L}\nonumber\\
& =0.\quad x\leq \hat{L}_{low}\nonumber\\
\psi_{2,2}(x)&=1,\quad x\geq v\nonumber\\
& =\dfrac{1}{q_2}(1-q_1-\dfrac{(f_1(v)-c)(1-q_1-q_2)}{f_1(x)-c})\quad \hat{L}>x>v\nonumber\\
& =0, \quad x\leq \hat{L}
\end{align}
It is easy to discern that the above strategy profile is a continuous distribution function. Also note that $\psi_{1,1}\neq \psi_{2,1}$ and $\psi_{1,2}\neq \psi_{2,2}$. Hence, the strategy profile is asymmetric. The following theorem confirms that the strategy profile that we have just described is indeed an NE.
\begin{thm}\label{thm:asymmne}
The strategy profile $\psi_{1}=\{\psi_{1,1}(\cdot),\psi_{1,2}(\cdot)\}$ and $\psi_2=\{\psi_{2,1}(\cdot),\psi_{2,2}(\cdot)\}$ as described in (\ref{eq:asymmne1}) and (\ref{eq:asymmne2}) respectively is an NE. 
\end{thm}
The above theorem confirms that there may exist an asymmetric NE when the penalty functions are of the form (\ref{mt1}).
\subsection{Strategy Profile Described in (\ref{c5}) may not be an NE}\label{sec:counterexample2}
We first describe a set of penalty functions that do not satisfy (\ref{con1}) and then we show that the strategy profile that we have characterized in (\ref{c5}) is not an NE. \footnote{Note that in theorem~\ref{thm4} we have shown that the strategy profile described in (\ref{c5}) is an NE when (\ref{con1}) is satisfied}

Consider that $n=2,c=0$; penalty functions are as follow:
\begin{align}
g_1(x)=x, f_1(x)=x\label{a3e0}\\
g_2(x)=\log(x), f_2(x)=e^x\label{a3e}
\end{align}
Now, as $c=0$, hence
\begin{align}\label{a3e1}
\dfrac{f_1(x)-c}{f_2(x)-c}=\dfrac{x}{e^x}=F(x)\nonumber\\
\dfrac{dF(x)}{dx}=e^{-x}-xe^{-x}
\end{align}
From (\ref{a3e1}), it is clear that for $x>1$, $F(x)$ is strictly decreasing. Hence we have for $1<x<y$
\begin{align}\label{a3e3}
\dfrac{f_2(y)-c}{f_1(y)-c}>\dfrac{f_2(x)-c}{f_1(x)-c}
\end{align}
which contradicts (\ref{con1}).\\
Now, let $v=5$ and $l=20,m=10, q_1=0.2,q_2=0.4$. Under this setting, we obtain
\begin{thm}\label{thm:notaNE}
The strategy profile as defined in (\ref{c5}) is \textbf{not} an NE strategy profile.
\end{thm}
\begin{rmk}
Thus, the condition in (\ref{con1}) is also {\em necessary} for a NE strategy profile to be in the form (\ref{c5}).
\end{rmk}
In the example constructed above, however, an NE strategy profile may still exist. Now, we show that in certain cases we can have a symmetric NE where $L_1<L_2$. Thus, when channel state is $1$ and $2$, a primary chooses penalty from the interval $[L_1,L_2]$ and $[L_2,v]$ respectively. (Note the difference; in the strategy profile described in lemma~\ref{lm:computation}, we have $L_1>L_2$).
Consider the following symmetric strategy profile where each primary selects strategy $\tilde{\psi}_i$ at channel state $i$, $i\in\{1,2\}$-
\begin{align}\label{a3e5}
\tilde{\psi}_i(x)= & 
 0 ,  \text{if } x<\tilde{L}_i\nonumber\\
& \dfrac{1}{q_i}(w^{-1}(\dfrac{f_i(x)-\tilde{p}_i}{f_i(x)-c})-\sum_{j=1}^{i-1}q_j), \text{if } \tilde{L}_{i+1}\geq x\geq \tilde{L}_i\nonumber\\
& 1,  \text{if } x>\tilde{L}_{i+1}
\end{align}
with $\tilde{L}_3=v$, $\tilde{L}_1>1$.\\
and 
\begin{align}\label{a3e2}
\tilde{p}_2& =(f_2(v)-c)(1-w(q_1+q_2))\nonumber\\
\tilde{L}_2& =g_2(\dfrac{\tilde{p}_2}{1-w(q_1)})\nonumber\\
\tilde{L}_1& =g_1(f_1(\tilde{L}_2)(1-w(q_1))\nonumber\\
\tilde{p}_1& =\dfrac{f_1(\tilde{L}_2)}{f_2(\tilde{L}_2)}*\tilde{p}_2
\end{align}
When $v=5$, $l=20,m=10, q_1=0.2, q_2=0.4$; we obtain $\tilde{p}_2=27.6185, \tilde{L}_2=3.3201, \tilde{p}_1=3.3148, \tilde{L}_1=3.3148$ .\\
It is easy to show that $\tilde{\psi}_i(\cdot)$ as defined in (\ref{a3e5}) is distribution function with $\tilde{\psi}_i(\tilde{L}_i)=0$ and $\tilde{\psi}_i(\tilde{L}_{i+1})=1$.  Note that under $\tilde{\psi}_i(\cdot)$, a primary selects penalty from the interval $[\tilde{L}_1,\tilde{L}_2]$ and $[\tilde{L}_2,v]$ when the channel states are $1$ and $2$ respectively. So, it remains to show that it is an NE strategy profile. The following theorem shows that it is indeed an NE strategy profile.
\begin{thm}\label{counterexample2}
Strategy profile $\tilde{\psi}_i(\cdot), i=1,\ldots,n$ as described in (\ref{a3e5}) is an NE .
\end{thm}

\section{Conclusion and Future Work}
We have analyzed a spectrum oligopoly market with primaries and secondaries where secondaries select a channel depending on the price quoted by a primary and the transmission rate a channel offers.  We have shown that in the one-shot game there exists a unique NE strategy profile which we have explicitly computed. We have analyzed the expected payoff under the NE strategy profile in the asymptotic limit and compared it with the payoff that primaries would obtain when they collude. We have shown that under a repeated version of the game there exists a subgame perfect NE strategy profile  where each primary obtains a payoff arbitrarily close to the payoff that it would have obtained if all primaries collude. 

The characterization of an NE strategy profile under the setting  i) when secondaries have different penalty functions and ii) when demand of secondaries vary depending on the pricing strategy remains an open problem. The analytical tools and results that we have provided may provide the basis for developing a framework to solve those problems.
\appendix
\vspace{-0.1cm}
\subsection{Proof of the results of Section~\ref{sec:singlestructure}}
We first state Lemma~\ref{lm:continuity},~\ref{lem:low_up_endpoint} and ~\ref{lm:disjoint} in order to prove Theorem~\ref{thm:noasymmetricNE}.  Theorem~\ref{thm1} readily follows from Lemma~\ref{lm:continuity}. After that we show Corollary~\ref{sup} which we use to prove theorems~\ref{thm2} and ~\ref{thm3}.

Now, we introduce some notations which we use throughout this section.
\begin{defn}
Let $r_i(x)$ denote the probability of winning of primary $i$ when it selects penalty $x$. \\
Let $t_{i}(x_1,\ldots,x_{i-1},x_{i+1},\ldots,x_l)$ denote the probability of at least $m$ success out of $l-1$ (except $i$) independent Bernouli event with event $k$ has success probability of $x_k$.
\end{defn}
Note that $r_i(x)$ does not depend on the state of the channel since secondaries select the channels based only on the penalties. Since secondaries prefer channels which induce lower penalty thus $r_i(\cdot)$ is non-increasing function. Note that $t_i(\cdot)$ is strictly increasing in each component. 
Note also that 
\begin{align}\label{eq:tiequalw}
t_i(x_1,\ldots,x_{i-1},x_{i+1},\ldots,x_l)=w(x) \quad (\text{if } x_1=\ldots=x_l=x).
\end{align}
If primary $i$ selects penalty $x$ at channel state $j$, then its expected payoff is
\begin{align}\label{eq:payoffi}
(f_j(x)-c)r_i(x).
\end{align}
\begin{defn}\label{br-asym}
A {\em Best response} penalty for primary $i$ at channel state $j\geq 1$ is 
\begin{align}
x=\underset{y\in \Re}{\text{argmax}}(f_j(y)-c)r_i(y)\nonumber
\end{align}
Let $u_{i,j,max}$ denote the maximum expected payoff under NE strategy profile for primary $i$ at channel state $j$ i.e. $u_{i,j,max}$ is equal to the payoff at $x$ at channel state $j$ if $x$ is a best response for primary $i$ at channel state $j$.
\end{defn}
A primary only selects penalty $x$ with positive probability at channel state $j$, if $x$ is a best penalty response at $j$. 
We state an observation that we will use throughout: \begin{obs}\label{o1}
Any penalty $y \leq g_j(c)$ can not be a best response (definition \ref{br}) for channel state $j$.
\end{obs}
\begin{proof}
Note that the profit of a primary is non-positive
if the selected penalty is upper bounded by $g_j(c)$. On the other hand when a primary selects penalty $x$ where $g_j(c)<x\leq v$, it can sell its channel at least in the event when the total number of available channels are less than $m$. Since  $0<\sum_{i=1}^{n}q_i<1$ by (\ref{prob}), thus the event occurs with non-zero probability, hence the profit is strictly positive when $g_j(c)<x\leq v$. Hence, the result follows.
\end{proof}
We denote $f(x-)=\lim_{y\uparrow x}f(y)$ throughout this section for a function $f(\cdot)$. Now we are ready to show Lemma~\ref{lm:continuity}.
\begin{lem}\label{lm:continuity}
$\psi_{i,j}(\cdot)$ is continuous at all points, except possibly at $v$. Also, at most one primary can have a jump at $v$.
\end{lem}
\begin{proof}
Let $\psi_{i,j}(\cdot)$ has a jump at $x<v$ where $i\in \{1,\ldots,l\}$ and $j\in \{1,\ldots,n\}$. Thus, $x$ is a best response for primary $i$ at channel state $j$. Next, we show that no player other than player $i$  will select penalty in the interval $[x,x+\epsilon]$ with positive probability for small enough $\epsilon>0$.

Fix a player $k\in \{1,\ldots,i-1,i+1,\ldots,l\}$ and channel state $k_1\in\{1,\ldots,n\}$. 

First, note that if $f_{k_1}(x)\leq c$, then a player can not select penalty in the interval $[x,x+\epsilon_0]$ with positive probability where $f_{k_1}(x+\epsilon_0)-c<(f_{k_1}(v)-c)(1-w(\sum_{j=1}^{n}q_j))$ since a primary gets a payoff of at least $(f_{k_1}(v)-c)(1-w(\sum_{j=1}^{n}q_i))$ at penalty $v$. Note that $\sum_{j=1}^{n}q_i<1$, thus $\epsilon_0>0$. We need to consider states $k_1\in\{1,\ldots,n\}$ such that $f_{k_1}(x)>c$. 

The payoff that player $k$ will get at a channel state $k_1$ at a $y\in [x,x+\epsilon_1]$ is 
\begin{align}\label{eq:payoffaty}
(f_{k_1}(y)-c)r_{k}(y)\leq (f_{k_1}(y)-c)r_{k}(x)\nonumber\\
(\text{since } r_k(y)\leq r_k(x)).
\end{align}

For any $\delta>0$, expected payoff for player $k$ at $x-\delta$ at channel state $k_1$ is lower bounded by
\begin{align}
(f_{k_1}(x-\delta)-c)r_{k}(x-) (\text{since } r_k(x-)\leq r_k(x-\delta)).
\end{align}
Since $\psi_{i,j}(\cdot)$ has a jump at $x$, thus $r_{k}(x-)>r_{k}(x)$. Hence, by continuity of $f_{k_1}(\cdot)$ there exists a $\delta>0$  and small enough $\epsilon_1>0$ such that for every $y\in [x,x+\epsilon_1]$  we have
\begin{align}
(f_{k_1}(x-\delta)-c)r_{k}(x-)& >(f_{k_1}(y)-c)r_{k}(x)\nonumber\\
& \geq (f_{k_1}(y)-c)r_{k}(x)\quad (\text{from } (\ref{eq:payoffaty})).\nonumber
\end{align}
Thus, player $k$ has strictly higher payoff at $x-\delta$ compared to at penalty $y\in [x,x+\epsilon_1]$. 

Hence, no player apart from $i$ selects penalty in the interval $[x,x+\epsilon]$ with positive probability where $\epsilon=\min(\epsilon_0,\epsilon_1)$. 

Since no player apart from player $i$ select penalty in the interval $[x,x+\epsilon]$, thus player $i$ will have a strictly higher payoff at $x+\epsilon$ instead of $x$ which contradicts the fact that $x$ is a best penalty response for player $i$.

If player $i$ selects $v$ with positive probability, then player $j,j\neq i$ will have strictly higher payoff by selecting a penalty just below $v$. Hence, player $j$ will select $v$ with $0$ probability. Hence, the result follows.
\end{proof}
We introduce some notations which we use throughout this section:
\begin{defn}
Let $X_{m,i}$ be the $m$th lowest penalty selected by primaries except $i$.
\end{defn}
Note that if primary $i$ selects penalty $x$, then it will not be able to sell its channel if $X_{m,i}<x$.
Now we show some results which directly follow from Lemma \ref{lm:continuity}. We use these results to prove Theorem~\ref{thm:noasymmetricNE}.
\begin{obs}\label{obs:successprob}
If $\psi_{k,j}(\cdot), k\neq i, j=1,\ldots,n$ does not have jump at $x$ apart from $i$ ($i$ may or may not have jump at $x$), then
\begin{align}\label{eq:penaltyequivalence}
r_i(x)=1-t_i(x_1,\ldots,x_{i-1},x_{i+1},\ldots,x_l) \nonumber\\ (\text{where } x_k=\sum_{j=1}^{n}q_j\psi_{k,j}(x)).
\end{align}
\end{obs}
\vspace{-0.3cm}
\begin{proof}
Note that 
\begin{align}\label{eq:successprob}
r_i(x)=P(A|X_{m,i}=x)P(X_{m,i}=x)+P(X_{m,i}>x)
\end{align}
where $P(A|X_{m,i}=x)$ is the probability that primary $i$ will be selected by secondaries when $X_{m,i}=x$. If $\psi_{k,j}(\cdot), k\neq i, j\in \{1,\ldots,n\}$ does not have any jump at $x$, then, 
\begin{align}\label{eq:nojump}
P(X_{m,i}=x)=0. 
\end{align}
On the other hand note that the probability of the event that  primary $k$ selects penalty less than or equal to $x$ is given by $\sum_{j=1}^{n}q_j\psi_{k,j}(x)$, hence, $P(X_{m,i}>x)$ is given by 
\begin{align}\label{eq:strictprob}
1-t_i(x_1,\ldots,x_{i-1},x_{i+1},\ldots,x_l)
\end{align}
where $x_k=\sum_{j=1}^{n}q_j\psi_{k,j}(x)$. Thus, if $\psi_{k,j}(\cdot)$ does not have jump at $x$ for all $k$ and $j$, then (\ref{eq:penaltyequivalence}) follows from (\ref{eq:successprob}), (\ref{eq:nojump}) and (\ref{eq:strictprob}).
\end{proof}

By Lemma~\ref{lm:continuity} no primary has a jump at $x<v$. Thus, $r_i(x)$ is exactly given by (\ref{eq:penaltyequivalence}) when $x<v$. By Lemma~\ref{lm:continuity}, only one player can have a jump at $v$. Thus, if player $k$ have a jump at $v$, then by Observation~\ref{obs:successprob}, $r_k(v)=r_k(v-)$. Hence, the following corollary is a direct consequence of (\ref{eq:penaltyequivalence}). 
\begin{cor}\label{obs:equality}
$r_{i}(x)=r_{k}(x)$ iff $\sum_{j=1}^{n}q_j\psi_{i,j}(x)=\sum_{j=1}^{n}q_j\psi_{k,j}(x)$ and $r_{i}(x)>r_{k}(x)$ iff $\sum_{j=1}^{n}q_j\psi_{i,j}(x)>q_j\sum_{j=1}^{n}\psi_{i,j}(x)$ for $x<v$ and for $x=v$ if no primary has a jump at $v$. If primary $k$ has a jump at $v$, then $r_k(v)=r_k(v-)$.
\end{cor}
\begin{defn}
Let 
\begin{eqnarray}
L_{i,j}=\inf\{x:\psi_{i,j}(x)>0\}.\label{eq:lowerendpoint}\\
U_{i,j}=\inf\{x:\psi_{i,j}=1\}.\label{eq:upperendpoint}
\end{eqnarray}
\end{defn}
$L_{i,j}, U_{i,j}$ are respectively the lowest and upper endpoint of support set of $\psi_{i,j}(\cdot)$. 
\begin{lem}\label{lem:low_up_endpoint}
i ) $L_{i,j}$ is a best response for primary $i$ at channel state $j$.\\
ii) $U_{i,j}$ is a best response for primary $i$ at channel state $j$, if one of the followings is true:\\
a) $U_{i,j}<v$. \\
b) $U_{i,j}=v$ and no primary has a jump at $U_{i,j}$.\\
c) $U_{i,j}=v$ and only primary $i$ has a jump at $v$.
\end{lem}
\begin{proof}
We prove part (i). The proof of part (ii) is similar and hence we omit.\\
We prove part (i) by considering the following two scenarios:\\
\textit{Case i}: $L_{i,j}=v$: Note that $\psi_{i,j}(v)=1$. Thus, by (\ref{eq:lowerendpoint}), $\psi_{i,j}(\cdot)$ has a jump at $L_{i,j}$; thus, $L_{i,j}$ is a best response to primary $i$ at channel state $j$. \\
\textit{Case ii}: $L_{i,j}<v$: By Lemma~\ref{lm:continuity}, no primary has a jump at $L_{i,j}$ and thus $r_i(\cdot)$ is continuous at $L_{i,j}$. Thus, by (\ref{eq:lowerendpoint}),  primary $i$ selects a penalty just above $L_{i,j}$ with positive probability when the channel state is $j$ i.e. for every $\epsilon>0$ there exists $y\in (L_{i,j},L_{i,j}+\epsilon)$ such that $y$ is a best response to primary $i$ at channel state $j$. Then, we must have a sequence $z_k, k=1,2,\ldots$ such that each $z_k$ is a best response for primary $i$ at channel state $j$ and $\underset{k\rightarrow \infty}{\text{lim}}z_k=L_{i,j}$. But profit to primary $i$ at channel state $j$ is $(f_j(z_k)-c)r_{i}(z_k)$. Now from continuity of $f_j(\cdot)$ and $r_i(\cdot)$ at $L_{i,j}$ we obtain
\begin{align}
\underset{k\rightarrow \infty}{\text{lim}}(f_j(z_k)-c)r_i(z_k)=(f_j(L_{i,j})-c)r_i(L_{i,j}).
\end{align}
Since each $z_k,k=1\ldots,\infty$ is a best response, thus, $L_{i,j}$ is also a best response.
\end{proof}
\begin{lem}\label{lm:disjoint}
$U_{i,a}\leq L_{i,j}$ if $a>j$
\end{lem}
\begin{proof}
Fix a primary $i$.  We first show that for any $x,y$ such that $x<y\leq v$, if $x$ is a best response when the state of the channel is $j$, then $y$ can not be a best response when the state of the channel is $a$ where $a>j$. If not, consider $y> x$ such that $x, y$ are the best responses when channel states are respectively  $j, a (a>j)$. Since $x$ is a best response at channel state $j$, thus $f_j(x)>c$ by Observation~\ref{o1}. On the other hand, since $y$ is a best response at channel state $a$, thus, $f_a(y)>c$ by Observation~\ref{o1}. Since $y>x$, thus $f_j(y)>c$.
  Also,
\begin{eqnarray}
u_{i,a,max} & = & (f_a(y)-c) r_{i}(y).\label{ptt1a} 
\end{eqnarray}
Expected payoff to primary $i$ at channel state $j$ at $y$ is $(f_j(y)-c)r_i(y)$. Thus, from (\ref{ptt1a})
\begin{align}\label{n41asym}
u_{i,j,max}\geq (f_j(y)-c)r_{i}(y)
     = u_{i,a,max}.\dfrac{f_j(y)-c}{f_a(y)-c}.
\end{align}
Since $x$ is a best response to primary $i$ at channel state $j$, thus
\begin{eqnarray}
u_{i,j,max} & = & (f_j(x)-c)r_{i}(x).\nonumber
\end{eqnarray}
Expected payoff of primary $i$ at channel state $a$ at penalty $x$ is
\begin{eqnarray}
(f_a(x)-c)r_{i}(x)
             = u_{i,j,max}.\dfrac{f_a(x)-c}{f_j(x)-c}\label{n42asym}
\end{eqnarray}
Using (\ref{n41asym}) and (\ref{n42asym}), we obtain-
\begin{eqnarray}\label{n43asym}
& (f_a(x)-c)r_{i}(x) \geq u_{i,a,max}.\dfrac{(f_j(y)-c)(f_a(x)-c)}{(f_a(y)-c)(f_j(x)-c)}\nonumber\\
& > u_{i,a,max}(\text{from } (\ref{con1}) \text{since } y>x, a>j, f_j(x)>c)
\end{eqnarray}
which contradicts Definition~\ref{br-asym}.

We also obtain from the argument in the above paragraph, if $U_{i,a}$ is a best response then $U_{i,a}<L_{i,j}$. 

 If $U_{i,a}$ is not a best response then by Lemma~\ref{lem:low_up_endpoint}, $U_{i,a}=v$ and there exists a primary other than $i$ which has a jump at $v$. Thus, by Lemma~\ref{lm:continuity}, primary $i$ does not have a jump at $v$. Thus, $\psi_{i,a}(\cdot)$ is continuous and thus, by the definition of $U_{i,a}$ (\ref{eq:upperendpoint}) for every $\epsilon>0$, there exists $y\in [v-\epsilon,v)$ such that $y$ is a best response for primary $i$ at channel state $a$. Hence, if $U_{i,a}>L_{i,j}$, then there exists a $y_1>L_{i,j}$ such that $y_1$ is a best response for primary $i$ for state $a$. But, we have already shown that it is not possible. Hence, the result follows.
\end{proof}

\textit{Proof of Theorem~\ref{thm:noasymmetricNE}}:
Suppose the statement is not true. Thus, we must have $i,k\in \{1,\ldots,l\}$ which do not have identical strategy. Let, $j$ be the largest index in $\{1,\ldots,n\}$ such that $\psi_{i,j}(\cdot)$ and $\psi_{k,j}(\cdot)$ differs.  Thus, we must have 
\begin{align}
x=\inf\{y\leq v: \psi_{i,j}(y)\neq \psi_{k,j}(y)\}.\nonumber
\end{align}
If $x=v$, then $\psi_{i,j}(\cdot)=\psi_{k,j}(\cdot)$ since $\psi_{i,j}(x^{\prime})=1$ for any $x^{\prime}\geq v$. Hence, we must have $x<v$.

Note that by definition of $j$, $\psi_{i,a}(x)=\psi_{k,a}(x)$ $\forall a> j$. Since $x<v$, thus $\psi_{i,j}(x)$ and $\psi_{k,j}(\cdot)$ are continuous at $x$ by Lemma~\ref{lm:continuity}, thus, $\psi_{i,j}(x)=\psi_{k,j}(x)$. Hence, $\psi_{i,a}(x)=\psi_{k,a}(x)$ for all $a\geq j$. By definition of $x$, $\psi_{i,j}(\cdot)$ and $\psi_{k,j}(\cdot)$ can not differ at a penalty less than $x$ and thus $\psi_{i,j}(x)=\psi_{k,j}(x)\neq 1$. Thus $x<U_{i,j}$ and $x<U_{k,j}$, hence $\psi_{i,a}(x)=\psi_{k,a}(x)=0$ $\forall a<j$ by Lemma~\ref{lm:disjoint}.
Since $\psi_{i,a}(x)=\psi_{k,a}(x)$ for all $a$ and $q_a, a=1,\ldots,n$ are exactly the same for each primary, thus, by Corollary~\ref{obs:equality} 
\vspace{-0.2cm}
\begin{align}\label{eq:samedist}
r_i(x)=r_k(x).
\end{align}
By definition of $x$, for every $\epsilon>0$, there is a $y$ such that $y\in (x,x+\epsilon)$ and $\psi_{i,j}(y)\neq \psi_{k,j}(y)$. Without loss of generality, we assume that $\psi_{i,j}(y)>\psi_{k,j}(y)$ for every $y$ in $(x,x+\epsilon)$ for some $\epsilon>0$. Thus, $\psi_{i,j}(x+\epsilon)>\psi_{i,j}(x)$ for every $\epsilon>0$. We consider the following two possible scenarios:

\textit{Case i}: $\psi_{k,j}(x+\epsilon)>\psi_{k,j}(x)$ for every $ \epsilon>0$. 

Hence, there is a $\gamma>0$, such that $x+\gamma\leq U_{k,j}$, $y\in (x,x+\gamma)$  is a best response for primary $k$ at channel state $j$ and $\psi_{i,j}(y)>\psi_{k,j}(y)$.  Since $\psi_{k,j}(x+\epsilon)>\psi_{k,j}(x)$ and $\psi_{i,j}(x+\epsilon)>\psi_{i,j}(x)$ for every $\epsilon>0$; and no primary has a jump at $x$, thus, $x$ is also a best response for primary $k$ and primary $i$ at channel state $j$. But, expected payoff to primary $k$ at $x$ at channel state $j$ is
\begin{align}\label{eq:same1}
(f_j(x)-c)r_{k}(x)=(f_j(x)-c)r_{i}(x)\quad(\text{from } (\ref{eq:samedist})).
\end{align}
Since $x$ is a best response for primary $i$ at channel state $j$, thus,
\begin{align}\label{eq:inequalpayoff}
(f_j(x)-c)r_i(x)\geq (f_j(y)-c)r_{i}(y).
\end{align}
Since $\psi_{i,j}(y)>\psi_{k,j}(y)$ and $\psi_{k,a}(y)=\psi_{k,a}(x)$ (since $L_{k,j}<y<U_{k,j}$) by Lemma~\ref{lm:disjoint} for all $a\neq j$, thus, $\sum_{a=1}^{n}q_a\psi_{i,a}(y)>\sum_{a=1}^{n}q_a\psi_{k,a}(y)$. Thus, from Corollary~\ref{obs:equality} $r_{k}(y)<r_{i}(y)$. Thus, expected payoff at $y$ is
\begin{align}\label{eq:same}
(f_j(y)-c)r_{k}(y)<(f_j(y)-c)r_{i}(y)\nonumber\\
\leq (f_j(x)-c)r_{i}(x)\quad (\text{from } (\ref{eq:inequalpayoff}))\nonumber\\
=(f_j(x)-c)r_{k}(x)\quad (\text{from } (\ref{eq:same1})).
\end{align}
Since $y$ and $x$ are best response to primary $k$ at channel state $j$, thus expected payoff to primary $k$ at channel state $j$ at $x$ and $y$ must be equal. But, this leads to a contradiction from (\ref{eq:same}) and (\ref{eq:same1}).

\textit{Case ii}: $\psi_{k,j}(x)=\psi_{k,j}(y)$ for some $y>x$: 

Let $x_1=\inf\{y> x:\psi_{k,j}(y)=\psi_{k,j}(x)\}$. Note that $\psi_{i,j}(x_1)>\psi_{k,j}(x_1)$. We consider two possible scenarios:

\textit{Case ii a}: $x_1<v$:

Since no primary has a jump at $x_1$ by Lemma~\ref{lm:continuity}, thus, by definition of $x_1$, it is a best response for primary $k$ at channel state $j$. But expected payoff to primary $k$ at channel state $j$ at $x_1$ is given by
\begin{align}
(f_j(x_1)-c)r_k(x_1).
\end{align}
 Since $\psi_{k,j}(x_1)=\psi_{k,j}(x)<1$, thus, $x_1<U_{k,j}$, thus, $\psi_{k,a}(x)=\psi_{k,a}(x_1)$ $\forall a<j$ by Lemma \ref{lm:disjoint}. Since $\psi_{i,j}(x+\epsilon)>\psi_{i,j}(x)$ for every $\epsilon>0$, thus, $x\geq L_{i,j}$, hence $\psi_{i,a}(x)=1$ $\forall a>j$. Since $x_1>x$, thus, $\psi_{i,a}(x_1)=\psi_{i,a}(x)=1$ for all $a>j$. If $\psi_{k,a}(x)<\psi_{k,a}(x_1)\leq 1$ for some $a>j$, then $\psi_{i,a}(\cdot)$ differs from $\psi_{k,a}(\cdot)$ (since $1=\psi_{i,a}(x)=\psi_{i,a}(x_1)$ for all $a>j$) at least at $x$, but this is against the definition of $j$. Hence, $\psi_{k,a}(x)=\psi_{k,a}(x_1)$ $\forall a$. Since $\psi_{i,j}(x_1)>\psi_{k,j}(x_1)$ and $\psi_{i,a}(x_1)\geq \psi_{i,a}(x)$ for all $a<j$,thus, $\sum_{a=1}^{n}q_a\psi_{i,a}(x)>\sum_{a=1}^{n}q_a\psi_{k,a}(x_1)$. Thus, by Corollary~\ref{obs:equality}, $r_{i}(x_1)>r_{k}(x_1)$. Since $x$ is a best response for player $i$ at channel state $j$, thus
\begin{align}\label{eq:paygreater}
(f_j(x)-c)r_{i}(x)\geq (f_j(x_1)-c)r_{i}(x_1) (\text{since } x_1>x).
\end{align}
Hence,
\begin{align}
(f_j(x)-c)r_{k}(x)& =(f_j(x)-c)r_{i}(x)\quad (\text{from } (\ref{eq:samedist}))\nonumber\\
& \geq (f_j(x_1)-c)r_{i}(x_1) \quad (\text{from } (\ref{eq:paygreater}))\nonumber\\
& >(f_j(x_1)-c)r_{k}(x_1).\nonumber
\end{align}
which contradicts the fact that $x_1$ is a best response for player $k$ when the channel state is $j$.

\textit{case ii b}: $x_1=v$:

  Since $\psi_{k,j}(v)$ must be $1$ and $\psi_{k,j}(v-)=\psi_{k,j}(x)<1$. Hence, $\psi_{k,j}(\cdot)$ must have a jump at $v$ and thus $v$ is a best response to primary $k$ when the channel state is $j$. Thus, by Corollary~\ref{obs:equality}, $r_k(v)=r_k(v-)$. Thus, by the continuity of $f_j(\cdot)$, penalties close to $v$ is also a best response for primary $k$ at channel state $j$  i.e. 
  \begin{align}\label{eq:sameclosev}
  (f_j(v)-c)r_k(v)=(f_j(v-)-c)r_k(v-).
  \end{align}
  Since $\psi_{k,j}(v-)=\psi_{k,j}(x)$ by the definition of $x_1$, thus, similar to argument in case ii a, we obtain $r_k(y)<r_i(y)$  for $x<y<v$, i.e. there exists an $\epsilon>0$, $r_i(v-\epsilon)>r_k(v-\epsilon)$ but $(f_j(v)-c)r_k(v)=(f_j(v-\epsilon)-c)r_k(v-\epsilon)$ (by (\ref{eq:sameclosev})). Thus, there exists an $\epsilon>0$,  such that  
  \vspace{-0.3cm}
  \begin{align}\label{eq:klessthani}
  (f_j(v-\epsilon)-c)r_k(v)=(f_j(v-\epsilon)-c)r_k(v-\epsilon)\nonumber\\
  <(f_j(v-\epsilon)-c)r_i(v-\epsilon).
  \vspace{-0.2cm}
  \end{align}
  Note that the right hand side is the expression for the expected payoff of primary $i$ at channel state $j$ at penalty $v-\epsilon$. Since $x$ is a best response to primary $i$ at channel state $j$, thus,
  \vspace{-0.2cm}
  \begin{align}
  (f_j(x)-c)r_k(x)& <(f_j(x)-c)r_i(x)\quad (\text{from } (\ref{eq:samedist}))\nonumber\\
  & \geq (f_j(v-\epsilon)-c)r_i(v-\epsilon)\nonumber\\
  & > (f_j(v )-c)r_k(v)\quad (\text{from } (\ref{eq:klessthani}))\nonumber
  \vspace{-0.2cm}
  \end{align}

  which contradicts that $v$ is a best response for primary $k$ at channel state $j$. Hence, $x_1\neq v$. Hence, this case does not arise.
  
  Thus from case i, case ii.a, and case ii.b we obtain the desired result. \qed
 
Henceforth, we denote $\psi_{i,j}$ and $r_i(\cdot)$ as $\psi_j(\cdot)$ and $r(\cdot)$ respectively by dropping the index corresponding to primary $i$. Note from Definition~\ref{defn:phi} that 
\begin{align}\label{ex1}
\phi_j(x)=(f_j(x)-c)r(x).
\end{align}
Also note that $L_{i,j}=L_{j}$ and $U_{i,j}=U_j$ for all $i\in \{1,\ldots,l\}$. Since strategy profiles of primaries are identical, thus, we can consider strategy profile in terms of only one primary (say, primary $1$). 

\textit{Proof of Theorem~\ref{thm1}}: By Lemma~\ref{lm:continuity} $\psi_i(\cdot)$ does not have a jump at $x<v$. If a primary has a jump at $v$, then by symmetric property other primaries also have a jump at $v$, which is not possible by Lemma~\ref{lm:continuity} since $l\geq 2$.  Thus, $\psi_i(\cdot)$ does not have a jump for any $i\in \{1,\ldots,n\}$.

Now, we show that $\phi_j(\cdot)$ is continuous. Now, we provide a closed form expression for $\phi_j(x)$ using (\ref{ex1}). Since $\psi_{i,j}=\psi_{j}(x)$, thus, from Observation~\ref{obs:successprob}, (\ref{eq:tiequalw}) and Theorem~\ref{thm1} the expected payoff for primary $i$ at $x$ at channel state $j$ is given by
\begin{align}
\phi_j(x)=(f_j(x)-c)(1-w(\sum_{k=1}^{n}q_j\psi_k(x))).\nonumber
\end{align}
The continuity follows from the equation due to the continuity of $w(\cdot)$ (Definition~\ref{defn:w}) and $\psi_k, k=1,\ldots,n$. \qed 

We obtain
\begin{align}\label{ex2}
\phi_j(x)=(f_j(x)-c)(1-w(\sum_{k=1}^{n}q_j\psi_k(x))).
\end{align}

Now, we show a corollary which is a direct consequence of Theorem~\ref{thm1}. We use this result to prove Theorems~\ref{thm2}  and \ref{thm3}. 

\begin{cor}\label{sup}
Every element in the support set of $\psi_i(\cdot)$ is a best response\footnote{Note that in general every element of a support set need not be a best response. To illustrate the fact consider a 3 player non co-operative game and the following NE strategy profile: players 1 and 2 have identical strategy profile which is a uniform distribution from $[L,v]$ and player 3 selects $v$ with probability $1$. Since the support set is closed, thus $v$ is in the support set for players $1$ and $2$. But, players $1$ and $2$ will attain strictly higher payoff at just below $v$ compared to at $v$. Thus, $v$ is not a best penalty response for players $1$ and$2$.  We show that this is not the case here because of the continuity of strategy profile. Specifically, a primary attains the highest possible expected payoff at every penalty in the support set in our setting.}; thus, so are $L_i, U_i$.
\end{cor}
\begin{proof}
Suppose that there exists a point $z$ in the support set of $\psi_i(\cdot)$, which is not a best response. Therefore, primary 1 plays at $z$ with probability $0$ when the channel state is $i$.\\Now,  one of the following two cases must arise.\\
\textit{Case I}: $\exists$ a neighborhood \cite{rudin} of radius $\delta> 0$ around $z$, such that no point in this neighborhood is a best response. Neighborhood of radius $\delta>0$ of $z$ is an open set (Theorem 2.19 of \cite{rudin}). Hence, we can eliminate that neighborhood and can attain a smaller closed set, such that its complement has probability zero under $\psi_i(\cdot)$, which contradicts the fact that $z$ is in the support set of $\psi_i(\cdot)$. \\
\textit{Case II}:  For every $\epsilon>0$, $\exists y\in (z-\epsilon,z+\epsilon)$, such that $y$ is a best response. Then, we must have a sequence $z_k, k=1,2,\ldots$  such that each $z_k$ is a best response, and $\underset{k\rightarrow \infty}{\text{lim}}z_k=z$ \cite{rudin}. But profit to primary 1 for channel state $i$  at each of $z_k$ is $(f_i(z_k)-c)(1-w(\sum_{j=1}^{n}q_j\psi_j(z_k)))$ by (\ref{ex2}). Thus, we can show that $z$ is also a best response from the continuity of $w(\cdot)$ ad $f_j(\cdot)$. We can conclude the result by noting that $U_i, L_i$ (Definition \ref{dlu}) are in the support set of $\psi_i(\cdot)$.
\end{proof}
\begin{rmk}
By corollary~\ref{sup}, at any channel state $i$, a primary attains the same expected payoff ($u_{i,max}$) at every point in the support set.
\end{rmk}
Now we are ready to prove theorems~\ref{thm2} and \ref{thm3}.

\textit{Proof of Theorem~\ref{thm2}}: Theorem~\ref{thm2} is the direct consequence of Lemma ~\ref{lm:disjoint} since $L_{i,j}=L_{j}$ and $U_{i,j}=U_{j}$ $\forall i$. \qed

\textit{Proof Of Theorem~\ref{thm3}}:
Suppose the statement is not true. But, it follows that there exists an interval $(x,y) \subseteq [L_n, v]$,
such that no primary offers penalty in the interval $(x,y)$ with positive probability. So, we must have $\tilde{a}$ such that
\begin{eqnarray}
\tilde{a}=\inf\{b\leq x:\psi_{j}(b)=\psi_{j}(x), \forall j\} .\nonumber
\end{eqnarray}

By definition of $\tilde{a}$, $\tilde{a}$ is a best response  for at least one state $i$. But, as no primary offers penalty in the range $(\tilde{a},y)$, so from \eqref{ex2}, $\phi_i(z) > \phi_i(\tilde{a})$ for
each $z \in (\tilde{a}, y)$. This is because $w(\sum_{j=1}^{n}q_j\psi_j(z))=w(\sum_{j=1}^{n}q_j\psi_j(\tilde{a}))$ and $ f_i(\tilde{a}) < f_i(z)$.  Thus, $\tilde{a}$ can not be  a best response for state i.\qed
\subsection{Proof of Results of Section~\ref{sec:computation}}\label{proofsec:computation}
We prove Lemma~\ref{lu},~\ref{lm:computation},~\ref{dcont}. Then, we state and prove  Observation~\ref{recurse} which we use to prove Theorem~\ref{thm4}.

\textit{Proof of Lemma~\ref{lu}}: We first outline a recursive computation strategy that leads to the expressions in (\ref{n51}) and (\ref{n52}).

Using Theorem~\ref{thm3}, we have $U_1=v$ and thus $v$ is a best response at channel state $1$ (Corollary~\ref{sup}). If a primary chooses penalty $v$ then it sells only when $X_m>v$, this allows us to compute $u_{1,max}$. By Theorem~\ref{thm3}, primaries with channel states $2,3,\ldots, n$ choose penalty below $L_1$ and primaries with channel state $1$ select penalty greater than $L_1$ with probability $1$. This allows us to calculate the payoff at $L_1$ which must be equal to $u_{1,max}$. The above equality allows us to compute the expression for $L_1$.

Since $L_1=U_2$ (Theorem~\ref{thm3}) and $U_2$ is a best response at channel state $2$, which enables us to obtain the expression for $u_{2,max}$. By Theorem~\ref{thm3} primaries with channel states $3,\ldots,n$ choose penalty below $L_2$ and primaries with channel state $1$ and $2$ select penalty greater than $L_2$ with probability $1$. This allows us to calculate the payoff at $L_2$ which must be equal to $u_{2,max}$. The above equality allows us to compute the expression for $L_2$. Using recursion, we can get the values of $u_{i,max}, L_i$, $ i=1,\ldots,n$. The detailed argument follows:

We first prove \eqref{n51} using induction, (\ref{n52}) follows from (\ref{n51}).

From theorem~\ref{thm3}, $\psi_i(\cdot)$'s support set is $[L_i,L_{i-1}]$ for $i=2,...,n$ and  $[L_1,v]$ for $i=1.$  Thus, $v$ is a best response for channel state  1 (by Corollary~\ref{sup}), hence
\vspace{-0.2cm}
\begin{align}
u_{1,max}& =(f_1(v)-c)(1-w(\sum_{i=1}^{n}q_i))= p_1-c.\nonumber
\vspace{-0.3cm}
\end{align}
Thus,  \eqref{n51} holds for $i=1$ with $L_0=v$.
 Let, \eqref{n51} be true for $i=t<n$. We have to show that \eqref{n51} is satisfied for $i=t+1$ assuming that it is true for $i=t$. Thus,by induction hypothesis,
 \vspace{-0.4cm}
\begin{align}\label{52a}
& u_{t,max}=p_t-c= (f_t(L_{t-1})-c)(1-w_{t}).
\vspace{-0.2cm}
\end{align}
Now, $L_t$ is a best response for state $t$, and thus,
\begin{align}
\phi_t(L_t) =  (f_t(L_t)-c)(1-w_{t+1})=p_t-c.\nonumber
\end{align}
Now, as $L_t $ is also a best response for state $t+1$ by Corollary~\ref{sup}, thus
\begin{eqnarray}
\phi_{t+1}(L_t) = (f_{t+1}(L_t)-c)(1-w_{t+1})=u_{t+1,max}.\nonumber
\end{eqnarray}
Thus, $u_{t+1,max}=p_{t+1}-c$ and it satisfies (\ref{n51}). Thus, \eqref{n51} follows from mathematical induction.

\eqref{n52} follows since
$(f_i(L_i)-c)(1-w_{i+1})=p_i-c$
and $g_i(\cdot)$ is the inverse of $f_i(\cdot).$\qed 

\textit{Proof of Lemma~\ref{lm:computation}}:
By Theorem~\ref{thm3}, $L_i, L_{i-1}$ are respectively the lower end-point and the upper end-point of the support set of $\psi_i(\cdot)$.
 We should have for $x<L_i$, $\psi_i(x)=0$ and for $x>L_{i-1}$,$\psi_i(x)=1$.  From Corollary ~\ref{sup},  every point $x \in [L_i,L_{i-1}]$ is a best response for state $i$, and hence,
\begin{align}
(f_i(x)-c)(1-w(\sum_{j=i+1}^{n}q_j+q_i.\psi_i(x)))=u_{i,max}=p_i-c. \nonumber
\end{align}
Thus, the expression for $\psi_i(\cdot)$  follows. We conclude the proof by noting that  the domain and range of $w(.)$ is $[0,1]$, and $\dfrac{p_i-c}{f_i(x)-c}<1$ for $x\in [L_i,L_{i-1}]$: so $w^{-1}(.)$ is defined at $1-\dfrac{p_i-c}{f_i(x)-c}$. \qed

\textit{Proof of Lemma~\ref{dcont}}:
Note that
\begin{align}
& \psi_i(L_i)=\dfrac{1}{q_i}(w^{-1}(1-\dfrac{p_i-c}{f_i(L_i)-c})-\sum_{j=i+1}^{n}q_j)\nonumber\\
           & =\dfrac{1}{q_i}(w^{-1}(w_{i+1})-\sum_{j=i+1}^{n}q_j)\quad \text{from}(\ref{n52})\nonumber\\
           & =0 \quad (\text{by} (\ref{d5})).\nonumber
\end{align}
From (\ref{c5}) and (\ref{n51}), we obtain
\begin{align}
& \psi_i(L_{i-1})=\dfrac{1}{q_i}(w^{-1}(1-\dfrac{p_i-c}{f_i(L_{i-1})-c})-\sum_{j=i+1}^{n}q_j)\nonumber\\
& =\dfrac{1}{q_i}(w^{-1}(w_i)-\sum_{j=i+1}^{n}q_j)\nonumber\\
& =\dfrac{1}{q_i}.q_i=1\quad (\text{as} w_i=w(\sum_{j=i}^{n}q_j)).\nonumber
\vspace{-0.3cm}
\end{align}
$w(.)$ is continuous, strictly increasing on compact set $[0, \sum_{j=1}^{n}q_j]$, so $w^{-1}$ is also continuous (theorem 4.17 in \cite{rudin}). Also, $\dfrac{p_i-c}{f_i(x)-c}$ is continuous for $x\geq L_i$ as $f_i(x)>c$, so $\psi_i(.)$ is continuous as it is a composition of two continuous functions. Again, $w^{-1}(.)$ is strictly increasing (as $w(\cdot)$ is strictly increasing), $1-\dfrac{p_i-c}{f_i(x)-c}$ is strictly increasing (as $f_i(\cdot)$ is strictly increasing), so $\psi_i(.)$ is strictly increasing on $[L_i,L_{i-1}]$ ( as it is a composition of two strictly increasing functions (Theorem 4.7 in \cite{rudin}))\qed.

Now, we state and prove a result (Observation~\ref{recurse}). Subsequently we prove Theorem~\ref{thm4}.

First, note that as $1-w_i>0,\forall i\in\{1,\ldots,n\}$, thus, $p_i-c>0$. Hence, from (\ref{n52}) it is evident that 
\begin{align}\label{imp}
f_k(L_k)>c.
\end{align}
\begin{obs}\label{recurse}
For $t>s, t,s\in\{1,\ldots,n\}$
\begin{align}\label{s32}
p_t-c=(p_s-c)\prod\limits_{i=s}^{t-1}\dfrac{f_{i+1}(L_i)-c}{f_i(L_i)-c}.
\end{align}
\end{obs}
\begin{proof}
Since $f_i^{-1}(\cdot)=g_i$, thus from (\ref{n52}) we obtain for $i-1$
\begin{align}\label{eq:multi}
p_{i-1}-c=(f_{i-1}(L_{i-1})-c)(1-w_{i}).
\end{align}
Hence, from (\ref{n51}), (\ref{imp}), and (\ref{eq:multi})
\begin{align}
p_{i}-c=(p_{i-1}-c)\dfrac{f_{i}(L_{i-1})-c}{f_{i-1}(L_{i-1})-c} .
\end{align}
We obtain the result using recursion.\end{proof}

\textit{Proof of Theorem \ref{thm4}}:
Fix a state $j\in \{1,\ldots,n\}$. First, we show that if a primary follows its strategy profile then it would attain a payoff of $p_j-c$ at channel state $j$. Next, we will show that if a primary unilaterally deviates from its strategy profile, then it would obtain a payoff of at most of $p_j-c$ (Case i and Case ii) when the channel state is $j$.

If the state of the channel of primary 1 is $i\geq 1$ and it selects penalty $x$, then its expected profit is-
\begin{align}\label{e1}
\phi_i(x)& =(f_i(x)-c)r(x)\nonumber\\
& =(f_i(x)-c)(1-w(\sum_{k=1}^{n}q_k.\psi_k(x))).
\end{align}
First, suppose $x\in [L_j,L_{j-1}]$. From (\ref{e1}) and (\ref{c5}), we obtain
\begin{align}\label{s43}
\phi_j(x)&= (f_j(x)-c)(1-w(\sum_{i=1}^{n}q_i\psi_i(x))) \notag\\& =(f_j(x)-c)(1-w(\sum_{k=j+1}^{n}q_k+q_j\psi_j(x)))\nonumber\\
 & =(f_j(x)-c)(1-w(w^{-1}(1-\dfrac{p_j-c}{f_j(x)-c})))\quad(\text{from } (\ref{c5}))\notag\\ 
 & =p_j-c.
\end{align}
Since $\psi_i(L_n)=0$ $\forall i$, we have
\begin{align}\label{paylow}
\phi_j(L_n)& =(f_j(L_n)-c)(1-w(0))=f_j(L_n)-c.
\end{align}
From (\ref{paylow}) expected payoff to a primary at state $j$ at $L_n$ is $f_j(L_n)-c$. At any $y<L_n$ expected payoff to a primary at state $j$ will be strictly less than $f_j(L_n)-c$. Hence, it suffices to show that for $x\in [L_k,L_{k-1}], k\neq j, k\in\{1,..,n\}$, profit to primary 1 is at most $p_j-c$, when the channel state is $j$. 

Now, let $x\in [L_k,L_{k-1}]$. From (\ref{e1}) and (\ref{c5}), expected payoff at $x$
\vspace{-0.3cm}
\begin{align}\label{ss2}
\phi_j(x)& =(f_j(x)-c)(1-w(\sum_{i=k+1}^{n}q_i+q_k\psi_k(x)))\nonumber\\
& =(f_j(x)-c)(1-w(w^{-1}(1-\dfrac{p_k-c}{f_k(x)-c})))(\text{from } (\ref{c5}))\notag\\ 
& =\dfrac{(p_k-c)(f_j(x)-c)}{f_k(x)-c}.
\end{align}
We will show that $\phi_j(x)-(p_j-c)$ is non-positive. As, $k\neq j$, so only the following two cases are possible.\\
\textit{Case i}:  $k<j$:
From (\ref{con1}), (\ref{imp}) and for $i<j$, we have-
\begin{equation}\label{ss4}
\dfrac{f_i(L_{i-1})-c}{f_i(L_i)-c}>\dfrac{f_j(L_{i-1})-c}{f_j(L_i)-c} (\text{as }  L_i<L_{i-1}).
\end{equation}
From Observation \ref{recurse} we obtain-
\begin{align}
p_j-c& = \dfrac{(p_k-c)(f_j(L_{j-1})-c)}{f_k(L_k)-c}\prod\limits_{i=k+1}^{j-1}\dfrac{f_i(L_{i-1})-c}{f_i(L_i)-c}.\nonumber
\end{align}
Using (\ref{ss4}) the above expression becomes
\begin{align}\label{eq:ss2a}
p_j-c& \geq \dfrac{(p_k-c)(f_j(L_{j-1})-c)}{f_k(L_k)-c}\prod\limits_{i=k+1}^{j-1}\dfrac{f_j(L_{i-1})-c}{f_j(L_i)-c}\nonumber\\
     & =\dfrac{(p_k-c)(f_j(L_k)-c)}{f_k(L_k)-c}.
\end{align}
Hence, from (\ref{ss2}) and (\ref{eq:ss2a}), we obtain-
\begin{align}\label{ss6}
& \phi_j(x)-(p_j-c)\notag\\& \leq (p_k-c)(\dfrac{f_j(x)-c}{f_k(x)-c}-\dfrac{f_j(L_k)-c}{f_k(L_k)-c}).
\end{align}
Since $x\in [L_k,L_{k-1}]$, $j>k$ and $f_k(L_k)>c$ (by (\ref{imp})); hence, from (\ref{ss6}) and  Assumption 1, we have-
\begin{eqnarray}\label{s5a}
\phi_j(x)\leq p_j-c.
\end{eqnarray}
\textit{Case ii}: $j<k$:
If $f_j(x)\leq c$ then a primary gets a non-positive payoff at channel state $j$, which is strictly below $p_j-c$. Hence we consider the case when $f_j(x)>c$. Since $x\leq L_{k-1}$ thus $f_j(L_{k-1})>c$.
Now, if $i>j$ and $f_j(L_i)>c$, we have from (\ref{con1}) and (\ref{imp})-
\begin{eqnarray}\label{s46}
\dfrac{f_{i}(L_{i-1})-c}{f_j(L_{i-1})-c}<\dfrac{f_{i}(L_{i})-c}{f_j(L_{i})-c}\quad (\text{as } L_i<L_{i-1})
\end{eqnarray}
Since $f_j(L_{k-1})>c$, thus 
\begin{align}\label{eq:greaterthan0}
f_j(L_i)>c\quad(\text{for } j\leq i<k, \text{as } L_i\geq L_{k-1}).
\end{align}
Now, from Observation~\ref{recurse} we obtain-
\begin{align}\label{eq:ss2b}
p_k-c& =(p_j-c)\prod\limits_{i=j}^{k-1}\dfrac{f_{i+1}(L_i)-c}{f_i(L_i)-c}\nonumber\\
       & =(p_j-c).\dfrac{f_k(L_{k-1})-c}{f_j(L_j)-c}\prod\limits_{i=j+1}^{k-1}\dfrac{f_{i}(L_{i-1})-c}{f_{i}(L_{i})-c}\nonumber\\
     & \leq (p_j-c).\dfrac{f_k(L_{k-1})-c}{f_j(L_j)-c}\prod\limits_{i=j+1}^{k-1}\dfrac{f_j(L_{i-1})-c}{f_j(L_{i})-c}\nonumber\\
     & (\text{from } (\ref{s46}), \& (\ref{eq:greaterthan0}))\nonumber\\
     & =(p_j-c).\dfrac{f_k(L_{k-1})-c}{f_j(L_{k-1})-c}.
\end{align}	
Thus, from (\ref{ss2}) and (\ref{eq:ss2b}), we obtain-
\begin{align}\label{s5}
& \phi_j(x)-(p_j-c)\notag\\& \leq (p_j-c)(\dfrac{f_k(L_{k-1})-c}{f_j(L_{k-1})-c}.\dfrac{f_j(x)-c}{f_k(x)-c}-1)\nonumber\\
   & \leq 0 (\text{as } x\leq L_{k-1}, j<k \quad\text{and from Assumption 1}).
\end{align}
Hence, from (\ref{s5}), (\ref{s5a}), and (\ref{s43}), every $x\in [L_j,L_{j-1}]$ is a best response to primary 1 when channel state is $j$. Since $j$ is arbitrary, it is true for any $j\in\{1,\ldots,n\}$ and thus (\ref{c5}) constitutes a Nash Equilibrium strategy profile.\qed
\vspace{-0.3cm}
\subsection{Proof of results of Section~\ref{sec:slnumerical}}
We first establish Lemma \ref{eff}. Subsequently, we prove Lemma \ref{thresh}.


\textit{Proof of Lemma \ref{eff}}: We divide the proof in three parts:
\begin{itemize}
\item First, we prove that when $m\geq (l-1)(\sum_{j=1}^{n}q_j+\epsilon)$ for some $\epsilon>0$, then $p_i-c\rightarrow f_i(v)-c$ as $l\rightarrow \infty$ (Part I).
\item Next we show that if $(l-1)\sum_{j=k}^{n}(q_j+\epsilon)\geq m\geq (l-1)\sum_{j=k+1}^{n}(q_j-\epsilon)$ for some $\epsilon>0$, then $p_i-c\rightarrow f_i(c_k)-c$ if $i>k$ and $p_i-c\rightarrow 0$ if $i\leq k$ (Part II).
\item Finally, we show if $m\leq (l-1)(q_n+\epsilon)$ for some $\epsilon>0$, then $p_i-c\rightarrow 0$ as $l\rightarrow \infty$ for any $i\in \{1,\ldots,n\}$ (Part III).
\end{itemize}
\textit{Part I}: Suppose  $m\geq (l-1)(\sum_{j=1}^{n}q_j+\epsilon)$ for some $\epsilon>0$.

Since $L_{i-1}\leq v$, thus, from (\ref{n51})-
\begin{align}\label{eq:indpayoff}
p_i-c\leq (f_i(v)-c) \quad i=1,\ldots,n.
\end{align}
When primary 1 selects penalty $v$ at channel state $i\geq 1$, then its expected profit is $\phi_i(v)=(f_i(v)-c)(1-w_1)$. Now,
from Theorem \ref{singlelocation} under the NE strategy profile,
\begin{align}\label{eq:indpayofflower}
p_i-c\geq \phi_i(v)=(f_i(v)-c)(1-w_1).
\end{align}
Let $Z_i, i=1,..,l-1$ be Bernoulli trials with success probabilities $\sum_{j=1}^{n}q_i$ and $Z=\sum_{i=1}^{l-1}Z_i$; so $P(Z\geq m)$ is equal to $w_1$ by (\ref{d5}).  Since $m\geq (l-1)(\sum_{i=1}^{n}q_i+\epsilon)$ for some $\epsilon>0$ and $E(Z)=(l-1)\sum_{i=1}^{n}q_i$, by weak law of large numbers \cite{ross}, $w_1\rightarrow 0$ as $l\rightarrow \infty$. Hence, $p_i-c\rightarrow f_i(v)-c$ as $l\rightarrow \infty$ by (\ref{eq:indpayoff}) and (\ref{eq:indpayofflower}).Thus, the result follows.\qed.

\textit{Part II}:
We show the result by evaluating the expressions for $p_j-c, j=1,\ldots,n$ in the asymptotic limit. Towards this end, we first evaluate the expressions for $w_j$ and $L_j$ in the asymptotic limit. We obtain the expression for $p_j-c$ when we combine those two values.\\
Suppose $(l-1)\sum_{j=k}^{n}(q_j+\epsilon)\geq m\geq (l-1)\sum_{j=k+1}^{n}(q_j-\epsilon)$ for some $\epsilon>0$. Since $w_{k+1}$ is the probability of at least $m$ successes out of $l-1$ independent Bernoulli trials, each of which occurs with probability $\sum_{j=k+1}^{n}q_j$ (by (\ref{d5})). Hence from the weak law of large numbers \cite{ross}
\begin{align}\label{asy1}
& 1-w_{k+1}\rightarrow 1\quad \text{as } l\rightarrow \infty.
\end{align}
Since $w_j<w_i$, for any $j>i$ (from (\ref{d5})), we have from (\ref{asy1}) for $j\geq k+1$
\begin{equation}\label{asy2}
1-w_j\rightarrow 1\quad \text{as } l\rightarrow \infty.
\end{equation}

Again, as $m\leq(l-1)(\sum_{j=k}^{n}q_j-\epsilon)$, so,  from weak law of large numbers\cite{ross},  for every $\epsilon>0$, $\exists L$, such that $1-w_{k}<\epsilon$, whenever $l\geq L$. Hence,
\begin{align}\label{asyz1}
1-w_{k}& \underset{l\rightarrow \infty}{\rightarrow} 0\nonumber\\
1-w_{j}& \underset{l\rightarrow \infty}{\rightarrow} 0 \quad (\text{for }j\leq k, w_j\geq w_k).
\end{align}
Thus, it is evident from (\ref{n51}) and (\ref{asyz1}) that if $i\leq k$, then
\begin{align}\label{asypaylo}
p_i-c\underset{l\rightarrow \infty}{\rightarrow} 0 .
\end{align}
Thus, from (\ref{n52}), (\ref{asy2}), and (\ref{asypaylo})
\begin{align}\label{asyl}
L_{k}\underset{l\rightarrow \infty}{\rightarrow} g_k(c)= c_{k}.
\end{align}
We obtain for $j> k$ from (\ref{n51})  and (\ref{asy2})
\begin{align}\label{asyimp}
p_j& \underset{l\rightarrow \infty}{\rightarrow} f_j(L_{j-1}).\quad (\text{from }(\ref{asy2})).
\end{align}
Again, using (\ref{n52}) and (\ref{asy2}), we obtain for $j> k$
\begin{align}\label{asyimp2}
p_j & \underset{l\rightarrow \infty}{\rightarrow} f_j(L_{j})\quad (\text{from }(\ref{asy2})).
\end{align}
$f_j(\cdot)$ is strictly increasing, thus from (\ref{asyimp}) and (\ref{asyimp2}), $L_j\rightarrow L_{j-1}$ (for $j>k$). Hence, for $j>k$,
\begin{align}\label{asylj}
L_j & \underset{l\rightarrow \infty}\rightarrow L_{k}\nonumber\\
L_j & \underset{l\rightarrow \infty}{\rightarrow} c_{k}\quad (\text{from}(\ref{asyl})).
\end{align}                               
Thus, from (\ref{asylj}), and  (\ref{asyimp2}), we obtain for any $i>k$ 
\begin{align}\label{asyhi}
p_i-c\underset{l\rightarrow \infty}{\rightarrow} (f_i(c_{k})-c).
\end{align}
Thus, from (\ref{asypaylo}) $p_i-c\rightarrow 0$ as $l\rightarrow \infty$ if $i\leq k$. From (\ref{asyhi}) we obtain $p_i-c\rightarrow f_i(c_{k})-c$ as $l\rightarrow \infty$ if $i>k$. Hence, the result follows.\qed

\textit{Part III}: 
Suppose that $m\leq (l-1)(q_n-\epsilon)$, for some $\epsilon>0$. Let, $Z_i, i=1,...,l-1$ be the Bernoulli trials with success probabilities $q_n$ and $Z=\sum_{i=1}^{l-1}Z_i$, $E(Z)=(l-1)q_n$. Hence,
\vspace{-0.3cm}
\begin{eqnarray}\label{e12}
1-w_n& \leq&  P(Z\leq m)\nonumber\\
& & \leq P(Z\leq (l-1)(q_n-\epsilon))\nonumber\\
& &   \leq P(|Z-(l-1)q_n|\geq (l-1)\epsilon)\nonumber\\
& & \leq 2\exp(-\dfrac{2(l-1)^2\epsilon^2}{l-1})\nonumber\\
& & (\text{from Hoeffding's Inequality \cite{hoeffding}})\nonumber\\
& & =2\exp(-2(l-1)\epsilon^2).
\vspace{-0.2cm}
\end{eqnarray}
Note that 
$1-w_i<1-w_j$ (if $j>i$), $f_k(L_{k-1})>f_{k-1}(L_{k-1})$. Hence, it can be readily seen from (\ref{n51}) that 
\vspace{-0.2cm}
\begin{align}\label{eq:pizero}
& p_i-c\leq (f_i(L_{i-1})-c)(1-w_n).
\vspace{-0.3cm}
\end{align}
Thus, the result follows from (\ref{e12}) and (\ref{eq:pizero}).\qed


When $m\leq (l-1)(q_n-\epsilon)$ for some $\epsilon>0$, then  the upper bound for $R_{NE}$  (see (\ref{eq:rne})) from (\ref{eq:pizero})is 
\begin{align}\label{nz}
R_{NE}& \leq (1-w_n)(\sum_{j=1}^{n}q_j.(f_j(L_{j-1})-c)).
\end{align}

Thus, for $m\leq (l-1)(q_n-\epsilon)$, $\epsilon>0$, from (\ref{e12}) and (\ref{nz}), we obtain.
\begin{equation}\label{exnz}
R_{NE}\leq \gamma\cdot\exp(-2\epsilon^2.(l-1))
\end{equation}
where $\gamma=2(1-w_n)(\sum_{j=1}^{n}q_j.(f_j(L_{j-1})-c))$. We will use this bound in proving Lemma~\ref{thresh}.

From, the definition of $\eta$, it should be clear that
\begin{equation}\label{bound}
\eta\leq 1.
\end{equation}

Now, we show Lemma \ref{thresh}

\textit{Proof of Lemma \ref{thresh}}: We divide the proof in the following two parts
\begin{itemize}
\item First, we show that  if $m\geq (l-1)(\sum_{i=1}^{n}q_i+\epsilon)$ for some $\epsilon>0$, then $\eta\rightarrow 1$ as $l\rightarrow \infty$ (Part I).
\item Next, we show that if  $m\leq (l-1)(q_n-\epsilon)$, for some $\epsilon>0$, then $\eta\rightarrow 0$ as $l\rightarrow \infty$ (Part II).
\end{itemize}
\textit{Part I}:
First suppose that $m\geq (l-1)(\sum_{i=1}^{n}q_i+\epsilon)$ for some $\epsilon>0$. 

From, definition of $R_{OPT}$, it is obvious that 
\begin{eqnarray}\label{e4}
R_{OPT}\leq l\cdot(\sum_{i=1}^{n}(q_i.(f_i(v)-c))).
\end{eqnarray}
Hence the result follows from Corollary \ref{cor:rne}, (\ref{e4}) and (\ref{bound}).\qed

\textit{Part II}:
Now, suppose that $m\leq (l-1)(q_n-\epsilon)$, for some $\epsilon>0$. We prove that $\eta\rightarrow 0$ as $l\rightarrow 0$. 

We prove the result by showing that $R_{NE}$ decreases at fast rate to $0$ compared to $R_{OPT}$ when $l\rightarrow\infty$.\\
 Let, $Z$ be the number of primaries, whose channel is in state $n$. Hence,
\begin{align}\label{boundropt}
& R_{OPT}\geq E(\min(Z,m))(f_n(v)-c)\nonumber\\& \dfrac{R_{OPT}}{f_n(v)-c}\geq E(\min(Z,m)).
\end{align}
Note that $E(Z)=l\cdot q_n$, $Var(Z)=l\cdot q_n(1-q_n)$.\\ We introduce a new random variable $Y$ as follows-
\begin{equation*}
Y=\begin{cases} m, &  \text{if} Z\geq m\\
0, & \text{otherwise.} \end{cases}
\end{equation*}
So, 
\begin{align}\label{e10}
E(\min(Z,m))& \geq  E(Y)\nonumber\\
& =m.P(Z\geq m)\nonumber\\
& \geq m.(1-P(Z\leq (l-1)(q_n-\epsilon)))\nonumber\\
& \geq m.(1-P(|Z-l.q_n|\geq (l-1)\epsilon)\nonumber\\
& \geq m.(1-\dfrac{l.q_n.(1-q_n)}{(l-1)^2.\epsilon^2}) \nonumber\\& (\text{From Chebyshev's Inequality}).
\vspace{-0.2cm}
\end{align}
Hence, from (\ref{exnz}), (\ref{boundropt}) and (\ref{e10}), we obtain-
\begin{eqnarray}\label{e8}
\eta\leq \dfrac{l.\gamma.\exp(-2(l-1)\epsilon^2)}{m.(1-\dfrac{l.q_n.(1-q_n)}{(l-1)^2.\epsilon^2}).(f_n(v)-c)} .\nonumber
\end{eqnarray}
Thus,  $\eta$ tends to zero for $m\leq (l-1)(q_n-\epsilon)$,as $l$ tends to infinity (as $m\neq 0$).\qed
\subsection{Proof of Results of Section~\ref{sec:assumpnecessary}}
First, we prove Lemma~\ref{lm:strategymultiple}. Subsequently, we state and prove Observation~\ref{obs:identity} which we use to show Theorem~\ref{thm:multipleNE}. The proof of Theorem ~\ref{thm:asymmne} is similar and hence we omit it. Finally, we show Theorems~\ref{thm:notaNE} and \ref{counterexample2}. 

\textit{Proof of Lemma~\ref{lm:strategymultiple}}:
First , it is evident from (\ref{mt3}) and (\ref{mt4})
\begin{align}
0\leq \dfrac{\bar{p}_1-c}{f_1(x)-c}\quad (\text{if} x\geq \bar{L})\nonumber
\end{align}
Hence, $w^{-1}(\cdot)$ is defined at $x\geq \bar{L}$. Note that
\begin{align}
\bar{\psi}(\bar{L})& =\dfrac{1}{\sum_{j=1}^{n}q_j}w^{-1}(1-\dfrac{\bar{p}_1-c}{f_1(g_1(\bar{p}_1))-c})\quad (\text{from} (\ref{mt4}))\nonumber\\
& =0\nonumber
\end{align}
 Note that
\begin{align}
\bar{\psi}(v)& =\dfrac{1}{\sum_{j=1}^{n}q_j}w^{-1}(1-\dfrac{\bar{p}_1-c}{f_1(v)-c})\nonumber\\
& =\dfrac{1}{\sum_{j=1}^{n}q_j}w^{-1}(1-\dfrac{(f_1(v)-c)(1-w_1)}{f_1(v)-c})\nonumber\\
& =\dfrac{1}{\sum_{j=1}^{n}q_j}w^{-1}(w_1)\nonumber\\
& =1\nonumber
\end{align}
We already know that $w^{-1}(\cdot)$ is continuous and strictly increasing. Since $1-\dfrac{\bar{p}_1-c}{f_1(x)-c}$ is strictly increasing and continuous for $x\geq \bar{L}$. Hence, $\bar{\psi}(\cdot)$ is continuous and strictly increasing on $[\bar{L},v]$. \qed.
\begin{obs}\label{obs:identity}
\begin{align}
f_i(\bar{L})=\bar{p}_i\quad (i=1,\ldots,n)
\end{align}
\end{obs}
\begin{proof}
The result is trivially true for $i=1$ by definition (\ref{mt4}) as $g_1(\cdot)=f_1^{-1}(\cdot)$. We will show the statement for $i\geq 2$. Since $f_i(\bar{L})>f_1(\bar{L})>c$, we have from (\ref{mt1})
\begin{align}
& \dfrac{f_1(\bar{L})-c}{f_i(\bar{L})-c}=\dfrac{f_1(v)-c}{f_i(v)-c}\nonumber\\
& \dfrac{\bar{p}_1-c}{f_i(\bar{L})-c}=\dfrac{f_1(v)-c}{f_i(v)-c}\nonumber\\
& \dfrac{(f_1(v)-c)(1-w_1)}{f_i(\bar{L})-c}=\dfrac{f_1(v)-c}{f_i(v)-c}\quad (\text{from} (\ref{mt3}))\nonumber\\
& f_i(\bar{L})-c=\bar{p}_i-c\quad (\text{from} (\ref{mt3}))\nonumber
\end{align}
Hence, the result follows.
\end{proof}
\textit{Proof of Theorem~\ref{thm:multipleNE}}:
We show that for any $x\in [\bar{L},v]$, a primary attains a payoff of $\bar{p}_i-c$ at channel state $i$. Then, we will show that if a primary selects a penalty outside the interval a primary\rq{}s payoff is strictly less than $\bar{p}_i-c$ at channel state $i$.

Suppose, $x\in[\bar{L},v]$. Now, fix any channel state $i\in \{1,\ldots,n\}$. If primary 1 selects penalty $x$ at channels state $i$, then its expected payoff is
\begin{align}\label{mt10}
\phi_i(x)& =(f_i(x)-c)(1-w(\sum_{j=1}^{n}q_j\bar{\psi}_j(x)))\nonumber\\
& =(f_i(x)-c)(1-w(\bar{\psi}(x)\sum_{j=1}^{n}q_j))\nonumber\\
& =(f_i(x)-c)(1-w(w^{-1}(1-\dfrac{\bar{p}_1-c}{f_1(x)-c})))\nonumber\\
& =\dfrac{(\bar{p}_1-c)(f_i(x)-c)}{f_1(x)-c}\nonumber\\
& =\dfrac{(f_1(v)-c)(f_i(x)-c)(1-w_1)}{f_1(x)-c}\nonumber\\& (\text{Using} (\ref{mt3}), f_1(\bar{L})>c)\nonumber\\
& =(f_i(v)-c)(1-w_1) \quad(\text{using} (\ref{mt1}))\nonumber\\
& =\bar{p}_i-c\quad(\text{from} (\ref{mt3}))
\end{align}
Now, at any $x<\bar{L}$, expected payoff will be strictly less than $f_i(\bar{L})-c$. But, from Observation ~\ref{obs:identity}, $f_i(\bar{L})-c=\bar{p}_i-c$. 

Thus, from (\ref{mt10}), when channel state is $i\geq 1$, every point in the interval $[\bar{L},v]$ is a best response to primary 1. Hence, the result follows.\qed 

\textit{Proof of Theorem~\ref{thm:notaNE}}:
 By simple calculation, we obtain the following values
\begin{equation*}
p_1=0.9305, L_1=1.1432, L_2=0.9372
\end{equation*}
Now, consider the following unilateral deviation for primary 1: primary 1 will choose a penalty $x\in (0.9305,0.9372)$ with probability 1, when the channel state is $1$. Since, no primary selects penalty lower than $0.9372$ under the strategy profile (\ref{c5}), thus expected payoff that primary 1 will obtain is $f_1(x)-c=x-c$, which is strictly larger than $p_1-c$ (since $x>0.9305=p_1$), the expected payoff that primary 1 gets by Theorem~\ref{singlelocation} when it selects strategy according to (\ref{c5}). Thus, the strategy profile as defined in (\ref{c5}) is not an NE.\qed 

\textit{Proof of Theorem~\ref{counterexample2}}:
We show that when the channel state is $i=2$ a primary does not have any profitable unilateral deviation from the strategy profile. The proof for $i=1$ is similar and thus we omit it.\\
First, we show that under the strategy profile a primary attains a payoff of $\tilde{p}_2-c$ at channel state $i=2$. Next, we show that if a primary deviates at channel state $2$, then its expected payoff is upper bounded by $\tilde{p}_2-c$.

Note that  when channel state is $2$ and primary chooses penalty $x\in [\tilde{L}_2,v] $, expected payoff to primary 1 is-
\vspace{-0.5cm}
\begin{align}\label{a3e10}
& (f_2(x)-c)(1-w(\sum_{i=1}^{2}q_i*\tilde{\psi}_i(x)))\nonumber\\& =(f_2(x)-c)(1-w((w^{-1}(\dfrac{f_2(x)-\tilde{p}_2}{f_2(x)-c})+q_1-q_1)))\nonumber\\
& =\tilde{p}_2-c
\end{align}
Now, suppose $x\in[\tilde{L}_1,\tilde{L}_2]$. From (\ref{a3e10}), expected payoff to a primary when it selects penalty $x$ at channel state $2$, is-
\begin{align}\label{a3e11}
& (f_2(x)-c)(1-w(q_1*\tilde{\psi}_1(x)))\nonumber\\& =(f_2(x)-c)(1-w(w^{-1}(\dfrac{f_1(x)-\tilde{p}_1}{f_1(x)-c})))\nonumber\\
& =(f_2(x)-c)\dfrac{\tilde{p}_1-c}{f_1(x)-c}\nonumber\\
& =\dfrac{(f_2(x)-c)(f_1(\tilde{L}_2)-c)}{(f_1(x)-c)(f_2(\tilde{L}_2)-c)}*(\tilde{p}_2-c)\nonumber\\& (\text{from (\ref{a3e2}) and $c=0$})\nonumber\\
& < \tilde{p}_2-c\quad (\text{from (\ref{a3e3}) as} \tilde{L}_2\geq x\geq \tilde{L}_1>1)
\end{align}
Note that at  $\tilde{L}_1$, $\tilde{\psi}_i(x)=0, i=1,2$. Hence, the expected payoff to a primary when it selects penalty $\tilde{L}_1$ at channel state $2$ is given by $f_2(\tilde{L}_1)-c$. From (\ref{a3e11}) we obtain $\tilde{p}_2>f_2(\tilde{L}_1)$, hence any penalty $x<\tilde{L}_1$ will induce payoff of strictly lower than $\tilde{p}_2$ when channel state is $2$. Hence, the result follows.\qed

\subsection{Proof of Results of Section~\ref{sec:repeatedgame}}
Here, we prove Theorem~\ref{spne}. Towards this end, we state and prove Observation~\ref{obsmono}.

First, we evaluate the total expected payoff that a primary will get under the strategy profile ($SP_{R}$). Note that the strategy $SP_{R}$ is symmetric, thus, the expected payoff of primaries would be identical and thus, we only evaluate the expected payoff of primary $1$. \\
Now we introduce some notations which we use throughout this section: 
\begin{defn}\label{dm}
 Let $X_{m}$ be the $m$th smallest offered penalty offered by  primaries $i=2,\ldots,l$.
\end{defn}
Let, $A_i$ denote the event that at a time slot, primary 1's channel will be bought, when its channel state is $i$ and selects penalty $v-\epsilon_i$ and primary $2,\ldots,l$ selects penalty $v-\epsilon_j$, when its channel state is $j$, $j\in\{1,\ldots,n\}$. Let's recall the definition of $X_m$ (definition ~\ref{dm}). From the law of total probability,
\begin{align}\label{prob1}
\Pr(A_i)& =\Pr(A_i|X_m>v-\epsilon_i)\Pr(X_m>v-\epsilon_i)\nonumber\\& +\Pr(A_i|X_m=v-\epsilon_i)\Pr(X_m=v-\epsilon_i)\nonumber\\& +\Pr(A_i|X_m<v-\epsilon_i)\Pr(X_m<v-\epsilon_i)
\end{align}
Now, note that $\Pr(A_i|X_m)=1 $ if $X_m>v-\epsilon_i$ and $\Pr(A_i|X_m)=0$ if $X_m<v-\epsilon_i$. 

Note from (\ref{condn2}) that
\begin{align}\label{eqw}
\Pr(X_m>v-\epsilon_i)& =\sum_{j=0}^{m-1}\dbinom{l-1}{j}(\sum_{k=i}^{n}q_k)^{j}(1-\sum_{k=i}^{n}q_k)^{l-1-j}\nonumber\\
& =1-w_i
\end{align}
Thus, the first term of right hand side (r.h.s.) of(\ref{prob1}) is $1-w_i$. We will denote the second term of the r.h.s. of (\ref{prob1}) as $\beta_i$. Since the third term of the r.h.s. of (\ref{prob1}) is zero, hence,
\begin{align}\label{winprob}
\Pr(A_i)=1-w_i+\beta_i
\end{align}
Thus, if primary follows the strategy profile as described, then its total expected payoff at any stage of the game will be
\begin{align}\label{resne}
R_{SNE}=\sum_{j=1}^{n}q_j(f_j(v-\epsilon_j))(1-w_j+\beta_j)
\end{align}
Next observation will be used in proving Theorem~\ref{spne}.
\begin{obs}\label{obsmono}
If a primary selects a penalty which is strictly greater than $v-\epsilon_i$, then the probability of winning is $\leq 1-w_i$.
\end{obs}
\begin{proof}
Consider that a primary selects penalty $x>v-\epsilon_i$. Note that only if $X_m\geq x$, then the channel of the primary may be bought\footnote{when $X_m=x$ then there is a nonzero probability that the channel may not be bought}. Hence, 
\begin{align}\label{eqshow}
\Pr(\text{the channel of the primary is bought})\leq \Pr(X_m\geq x)
\end{align}
Since, $x>v-\epsilon_i$, thus
\begin{align}\label{boundprob}
& \Pr(X_m>v-\epsilon_i)=\Pr(X_m\geq x)+\Pr(x>X_m>v-\epsilon_i)\nonumber\\
& =>\Pr(X_m>v-\epsilon_i)\geq \Pr(X_m\geq x)
\end{align}
Hence, using (\ref{boundprob}) in (\ref{eqshow}), we obtain
\begin{align}\label{re1}
\Pr(\text{the channel of the primary is bought})\leq \Pr(X_m>v-\epsilon_i)
\end{align}
But from (\ref{eqw}), $1-w_i=\Pr(X_m>v-\epsilon_i)$, hence the result follows from (\ref{re1}).
\end{proof}

Now, we are ready to  prove Theorem~\ref{spne}. 

\textit{proof of Theorem~\ref{spne}}:
Fix any state $i$. We prove the theorem in two part. In Part 1, we show that when the game is at a stage where other primaries select penalty $v-\epsilon_j, j=1,\ldots,n$ at channel state $j$, then primary 1 does not have any unilateral deviation by selecting penalty different from $v-\epsilon_i$ for sufficiently high $\delta$.  In part 2, we show that if  the game is in a stage where all the other primaries play the unique NE strategy profile, then primary 1 also does not have any profitable unilateral deviation. This will ensure that $SP_R$ is a subgame perfect NE.\\
\textit{ Proof of Part 1}: First, we show that, deviating to lower penalty compared to $v-\epsilon_i$ is not profitable (case 1) and then we show that deviating to a higher penalty compared to $v-\epsilon_i$, is also not profitable (case 2) when other primaries select penalty $v-\epsilon_j, j=1,\ldots,n$ at channel state $j$.

\textit{Case 1}: First, suppose that primary 1 offers penalty, which is strictly less than $v-\epsilon_i$.

A primary can attain at most a payoff of $f_i(v-\epsilon_i)$ at this stage. After this deviation, all the primaries play the unique N.E. strategy. The payoff is given by $R_{NE}$. Now,  
\begin{align}\label{r9}
R_{NE}& =\sum_{j=1}^{n}q_j(p_j-c)\nonumber\\
& =\sum_{j=1}^{n}q_j(f_j(L_{j-1})-c)(1-w_j)\quad(\text{from }(\ref{n51}))\nonumber\\
& \leq \sum_{j=1}^{n}q_j(f_j(v-\epsilon_j)-c)(1-w_j)\quad(\text{from} (\ref{condn3}))
\end{align}
 If a primary deviates at stage $T$, then its expected payoff starting from stage $T$ would be at most
\begin{align}\label{r30}
& (1-\delta)[\delta^{T}(f_i(v-\epsilon_i)-c)+\sum_{t=T+1}^{\infty}\delta^t R_{NE}]\nonumber\\
& =\delta^{T}[(1-\delta)(f_i(v-\epsilon_i)-c)+\delta R_{NE}]
\end{align}
If a primary would not have deviated, then its expected payoff would have been
\begin{align}\label{r31}
& (1-\delta)[\delta^{T}(f_i(v-\epsilon_i)-c)(1-w_i+\beta_i)+\sum_{t=T+1}^{\infty}\delta^t R_{SNE}]\nonumber\\
&= \delta^{T}[(1-\delta)(f_i(v-\epsilon_i)-c)(1-w_i+\beta_i)+\delta R_{SNE}]
\end{align}
Hence, from  (\ref{r30}) and (\ref{r31}), the following condition must be satisfied for sub game perfect equilibrium
\begin{align}\label{r1}
& \delta^{T}[(1-\delta)(f_i(v-\epsilon_i)-c)+\delta R_{NE}]\leq\nonumber\\ & \delta^{T}[(1-\delta)(f_i(v-\epsilon_i)-c)(1-w_i+\beta_i)+\delta R_{SNE}] 
\end{align}
From (\ref{resne}) and (\ref{r9}), it is enough to satisfy the following inequality in order to satisfy inequality (\ref{r1})
\begin{align}\label{r10}
& (1-\delta)(f_i(v-\epsilon_i)-c)+\delta\sum_{j=1}^{n}q_j(f_j(v-\epsilon_j)-c)(1-w_j)\leq\nonumber\\& (1-\delta)(f_i(v-\epsilon_i)-c)(1-w_i+\beta_i)\nonumber\\& +\delta\sum_{j=1}^{n}q_j(f_j(v-\epsilon_j)-c)(1-w_j+\beta_j)
\end{align}
By simple algebraic manipulation in (\ref{r10}), we obtain
\begin{align}\label{r10a}
\delta \geq \dfrac{(f_i(v-\epsilon_i)-c)(w_i-\beta_i)}{(f_i(v-\epsilon_i)-c)(w_i-\beta_i)+\sum_{j=1}^{n}q_j(f_j(v-\epsilon_j)-c)\beta_j}
\end{align}
The proof is complete by observing that the right hand side of (\ref{r10a}) is strictly less than 1. Hence, if $\delta$ is greater than the following expression
\begin{equation*}
\max_{i\in{1,\ldots,n}}\dfrac{(f_i(v-\epsilon_i)-c)(w_i-\beta_i)}{(f_i(v-\epsilon_i)-c)(w_i-\beta_i)+\sum_{j=1}^{n}q_j(f_j(v-\epsilon_j)-c)\beta_j}
\end{equation*}
then, a primary will not have any profitable one shot deviation.

\textit{Case 2}: Now, suppose that primary 1 offers penalty which is strictly greater than $v-\epsilon_i$. 

Since primary 1 offers penalty strictly greater than  $v-\epsilon_i$ and $v-\epsilon_i\leq v$, thus, from observation ~\ref{obsmono}, primary 1 can at most attain a payoff of $(f_i(v)-c)(1-w_i)$ by offering penalty higher than $v-\epsilon_i$. After the deviation, all the primaries play the one-shot NE strategy profile.
If a primary deviates at stage $T$, then its expected payoff starting from stage $T$ would be
\begin{align}\label{r30a}
& (1-\delta)[\delta^{T}(f_i(v)-c)(1-w_i)+\sum_{t=T+1}^{\infty}\delta^t R_{NE}]\nonumber\\
& =\delta^{T}[(1-\delta)(f_i(v-\epsilon_i)-c)+\delta R_{NE}]
\end{align}
If a primary would not have deviated, then its expected payoff would be
\begin{align}\label{r31a}
& (1-\delta)[\delta^{T}(f_i(v-\epsilon_i)-c)(1-w_i+\beta_i)+\sum_{t=T+1}^{\infty}\delta^t R_{SNE}]\nonumber\\
&= \delta^{T}[(1-\delta)(f_i(v-\epsilon_i)-c)(1-w_i+\beta_i)+\delta R_{SNE}]
\end{align}
Hence, from (\ref{r30a}) and (\ref{r31a}), the following condition must be satisfied for sub game perfect equilibrium
\begin{align}\label{r1a}
& \delta^{T}[(1-\delta)(f_i(v)-c)(1-w_i)+\delta R_{NE}]\leq\nonumber\\ & \delta^{T}[(1-\delta)(f_i(v-\epsilon_i)-c)(1-w_i+\beta_i)+\delta R_{SNE}] 
\end{align}
Hence, from (\ref{resne}) satisfying the following condition will be enough for the strategy profile to be a SPNE.
\begin{align}\label{r11}
& (1-\delta)(f_i(v)-c)(1-w_i)+\delta R_{NE}\nonumber\\& \leq (1-\delta)(f_i(v-\epsilon_i)-c)(1-w_i+\beta_i)\nonumber\\& +\delta\sum_{j=1}^{n}q_j(f_j(v-\epsilon_j)-c)(1-w_j+\beta_j)
\end{align}
But, from (\ref{condn3})
\begin{equation*}
(f_i(v)-c)(1-w_i)\leq (f_i(v-\epsilon_i)-c)
\end{equation*}
Hence, it is enough to satisfy the following inequality in order to satisfy inequality (\ref{r11})
\begin{align}
& (1-\delta)(f_i(v-\epsilon_i)-c)+\delta\sum_{j=1}^{n}q_j(f_j(v-\epsilon_j)-c)(1-w_j)\leq \nonumber\\& (1-\delta)(f_i(v-\epsilon_i)-c)(1-w_i+\beta_i)\nonumber\\& +\delta\sum_{j=1}^{n}q_j(f_j(v-\epsilon_j)-c)(1-w_j+\beta_j)
\end{align}
which is exactly similar to (\ref{r10}). 
Hence, the rest of the proof will be similar to case 1.

\textit{proof of part 2}: If other primaries play the unique one-shot NE strategy, selecting the unique one-shot NE strategy is a best response for a primary. Hence, no primary has any profitable unilateral deviation. 


\begin{thebibliography}{10}
\bibitem{estimation}
S.E. Bensley and B.~Aazhang.
\newblock Subspace-based channel estimation for code division multiple access
  communication systems.
\newblock {\em Communications, IEEE Transactions on}, 44(8):1009--1020, Aug
  1996.

\bibitem{cover}
Thomas~M. Cover and Joy~A. Thomas.
\newblock {\em Elements of Information Theory}.
\newblock Wiley, 2nd edition, 2006.

\bibitem{duan}
Lingjie Duan, Jianwei Huang, and Biying Shou.
\newblock Competition with dynamic spectrum leasing.
\newblock In {\em New Frontiers in Dynamic Spectrum, 2010 IEEE Symposium on},
  pages 1--11, April 2010.

\bibitem{df}
B.S. Everitt.
\newblock {\em {The Cambridge Dictionary of Statistics}}.
\newblock 3rd Edition, Cambridge University Press, 2006.

\bibitem{Fey}
Mark Fey.
\newblock Symmetric games with only asymmetric equilibria.
\newblock {\em Games and Economic Behavior}, 75(1):424 -- 427, 2012.

\bibitem{fischer}
R.F.H. Fischer and J.B. Huber.
\newblock A new loading algorithm for discrete multitone transmission.
\newblock In {\em Global Telecommunications Conference, 1996. GLOBECOM '96.},
  volume~1, pages 724--728 vol.1, Nov 1996.

\bibitem{isit}
A.~Ghosh and S.~Sarkar.
\newblock {Quality sensitive price competition in spectrum oligopoly}.
\newblock In {\em Proceedings of IEEE International Symposium on Information
  Theory (ISIT)}, pages 2770--2774, 2013 (Proofs are available at
  http://arxiv.org/abs/1305.3351).

\bibitem{archived-part1}
Arnob Ghosh and Saswati Sarkar.
\newblock Quality sensitive price competition in spectrum oligopoly: Part 1.
\newblock {\em CoRR}, abs/1404.2514, 2014.

\bibitem{hoeffding}
W.~Hoeffding.
\newblock {Probability inequalities for sums of bounded random variables}.
\newblock {\em Journal of the American Statistical Institute}, 58, 1963.

\bibitem{Ileri}
O.~Ileri, D.~Samardzija, T.~Sizer, and N.~B. Mandayam.
\newblock {Demand Responsive Pricing and Competitive Spectrum Allocation via a
  Spectrum Policy Server}.
\newblock In {\em IEEE Proceedings of DySpan}, pages 194--202, 2005.

\bibitem{Janssen}
M.~Janssen and E.~Rasmusen.
\newblock {Bertrand Competition Under Uncertainty}.
\newblock {\em Journal of Industrial Economics}, 50(1):11--21, March 2002.

\bibitem{jia}
Juncheng Jia and Qian Zhang.
\newblock Bandwidth and price competitions of wireless service providers in
  two-stage spectrum market.
\newblock In {\em Communications, 2008. ICC '08. IEEE International Conference
  on}, pages 4953--4957, May 2008.

\bibitem{Gaurav1}
G.S. Kasbekar and S.~Sarkar.
\newblock {Spectrum Pricing Game with arbitrary Bandwidth Availability
  Probabilities}.
\newblock In {\em Proceeding of IEEE International Symposium on Information
  Theory (ISIT)}, pages 2711--2715, 2011.

\bibitem{gauravjsac}
G.S. Kasbekar and S.~Sarkar.
\newblock {Spectrum Pricing Game with Bandwidth Uncertainty and Spatial Reuse
  in Cognitive Radio Network}.
\newblock {\em IEEE Journal on Special Areas in Communication}, 30(1):153--164,
  2012.

\bibitem{kavurmacioglu}
Alanyali~M. Kavurmacioglu, E. and D.~Starobinski.
\newblock Competition in secondary spectrum markets: Price war or market
  sharing?
\newblock In {\em IEEE Proceedings of DYSPAN}.

\bibitem{kim}
Hyoil Kim, Jaehyuk Choi, and K.G. Shin.
\newblock Wi-fi 2.0: Price and quality competitions of duopoly cognitive radio
  wireless service providers with time-varying spectrum availability.
\newblock In {\em INFOCOM, 2011 Proceedings IEEE}, pages 2453--2461, April
  2011.

\bibitem{Kimmel}
S.~Kimmel.
\newblock {Bertrand Competition Without Completely Certain Productions}.
\newblock {\em Economic Analysis Group Discussion Paper, Antitrust Division,
  U.S. Department of Justice}, 2002.

\bibitem{Kreps}
D.M. Kreps and J.A. Scheinkman.
\newblock {Quantity Precommitment and Bertrand Competition yield Cournot
  Outcomes}.
\newblock {\em Bell Journal of Economics}, 14:326--337, Autumn, 1983.

\bibitem{lin}
Peng Lin, Juncheng Jia, Qian Zhang, and M.~Hamdi.
\newblock Dynamic spectrum sharing with multiple primary and secondary users.
\newblock {\em Vehicular Technology, IEEE Transactions on}, 60(4):1756--1765,
  May 2011.

\bibitem{Luo}
Zhi-Quan Luo and Shuzhong Zhang.
\newblock Dynamic spectrum management: Complexity and duality.
\newblock {\em Selected Topics in Signal Processing, IEEE Journal of},
  2(1):57--73, Feb 2008.

\bibitem{Mailepricecompslotted}
P.~Maille and B.~Tuffin.
\newblock {Analysis Of Price Competition in a Slotted Resource Allocation
  Game}.
\newblock In {\em Proceeding of 27th IEEE INFOCOM}, pages 888--896, 2008.

\bibitem{Mailespectrumsharing}
P.~Maille and B.~Tuffin.
\newblock {Price War with Partial Spectrum Sharing for Competitive Wireless
  Service Provider}.
\newblock In {\em Proceeding Of IEEE GLOBECOM}, pages 1--6, 2009.

\bibitem{mwg}
A.~Mas~Colell, M.~Whinston, and J.~Green.
\newblock {\em { Microeconomic Theory}}.
\newblock Oxford University Press, 1995.

\bibitem{Niyatospeccrn}
D.~Niyato and E.~Hossain.
\newblock {Competitive Pricing for Spectrum Sharing in Cognitive Radio Network:
  Dynamic Games, Inefficiency of Nash Equilibrium, and Collusion}.
\newblock {\em IEEE Journal on Special Areas in Communication}, 26(1):192--202,
  2008.

\bibitem{Niyatomultipleseller}
D.~Niyato, E.~Hossain, and Z.~Han.
\newblock {Dynamics of Multiple Seller and Multiple Buyer Spectrum Trading in
  Cognitive Radio Network: A Game theoretic Modeling approach}.
\newblock {\em IEEE Transaction on Mobile Computing}, 8(8):1009--1022, 2009.

\bibitem{Osborne}
M.J. Osborne and C.~Pitchik.
\newblock {Price Competition in a Capacity Constrained Duopoly}.
\newblock {\em Journal On Economic Theory}, 38(2):238--260, 1986.

\bibitem{ross}
Sheldon Ross.
\newblock {\em A First Course in Probability}.
\newblock 8th Edition, Prentice Hall, 2009.

\bibitem{rudin}
W.~Rudin.
\newblock {\em {Principles Of Mathematical Analysis}}.
\newblock Third Edition, Mc-Graw Hill, 1976.

\bibitem{sengupta}
S.~Sengupta and M.~Chatterjee.
\newblock An economic framework for dynamic spectrum access and service
  pricing.
\newblock {\em Networking, IEEE/ACM Transactions on}, 17(4):1200--1213, Aug
  2009.

\bibitem{song}
G.~Song and Ye~Li.
\newblock Utility-based resource allocation and scheduling in ofdm-based
  wireless broadband networks.
\newblock {\em Communications Magazine, IEEE}, 43(12):127--134, Dec 2005.

\bibitem{yitan}
Yi~Tan, S.~Sengupta, and K.P. Subbalakshmi.
\newblock Competitive spectrum trading in dynamic spectrum access markets: A
  price war.
\newblock In {\em Global Telecommunications Conference (GLOBECOM 2010), 2010
  IEEE}, pages 1--5, Dec 2010.

\bibitem{Xing}
Y.~Xing, R.~Chandramouli, and C.~Cordeiro.
\newblock {Price Dynamics in Competitive Agile Spectrum Access Markets}.
\newblock {\em IEEE Journal on Special Areas in Communication}, 25(3):613--621,
  2008.

\bibitem{Xu}
Hong Xu, Jin Jin, and Baochun Li.
\newblock A secondary market for spectrum.
\newblock In {\em INFOCOM, 2010 Proceedings IEEE}, March 2010.

\bibitem{yang}
Lei Yang, Hongseok Kim, Junshan Zhang, Mung Chiang, and Chee wei Tan.
\newblock Pricing-based spectrum access control in cognitive radio networks
  with random access.
\newblock In {\em INFOCOM, 2011 Proceedings IEEE}, pages 2228--2236, April
  2011.

\bibitem{zhang}
Feng Zhang and Wenyi Zhang.
\newblock Competition between wireless service providers: Pricing, equilibrium
  and efficiency.
\newblock In {\em Modeling Optimization in Mobile, Ad Hoc Wireless Networks
  (WiOpt), 2013 11th International Symposium on}, pages 208--215, May 2013.

\bibitem{Zhou}
X.~Zhou and H.~Zheng.
\newblock {TRUST: A General Framework for Truthful Double Spectrum Auctions}.
\newblock In {\em In the Proceedings of Infocom}, April 2009.

\end{thebibliography}
\end{document}